\documentclass{article} 
\usepackage{iclr2026_conference,times}

\usepackage[T1]{fontenc}
\usepackage[utf8]{inputenc}  
\usepackage{hyperref}
\usepackage{url}
\usepackage{amsmath,amssymb,amsthm}
\usepackage{bm}
\usepackage{booktabs}
\usepackage{enumitem}
\usepackage{graphicx}
\usepackage{microtype}
\usepackage{titletoc}
\usepackage{multirow}
\newtheorem{assumption}{Assumption}
\newtheorem{lemma}{Lemma}
\newtheorem{theorem}{Theorem}
\newtheorem{proposition}{Proposition}
\newtheorem{corollary}{Corollary}
\newtheorem{remark}{Remark}

\newcommand{\Eold}{\mathbb{E}_{\bm a\sim\pi_{\bm\theta_{\mathrm{old}}}}}
\newcommand{\R}{\mathrm{R}}
\newcommand{\Var}{\mathrm{Var}}
\newcommand{\one}{\mathbf 1}
\newcommand{\SA}{S^{A}}
\newcommand{\SR}{S^{R}}
\newcommand{\tS}{\widetilde S}

\title{Aggregate in the Advantage, Not the Ratio:\\
A Canonical-Form Analysis of Cooperative Multi-Agent Policy Optimization}

\author{Zijian Zhao$^1$, Sen Li$^{1,2}$ \thanks{Corresponding Author: Sen Li} \\
$^1$The Hong Kong University of Science and Technology \\
$^2$The Hong Kong University of Science and Technology (Guangzhou)
}

\begin{document}
\maketitle

\begin{abstract}
Multi-agent policy optimization, exemplified by PPO-based methods, is a key branch of cooperative Multi-Agent Reinforcement Learning (MARL). A central design question is how many neighboring agents\footnote{In this paper, "neighbors" refer not only to physical proximity but also to agents whose actions influence one another.} to aggregate in order to effectively utilize global information for cooperation. This decision must be made along two dimensions: in the \textbf{advantage} (which agents' rewards contribute to the credit signal) and in the \textbf{ratio} (which agents' likelihood ratios form the clipped importance weight). Existing methods occupy scattered, underexplored points on these two axes: IPPO treats both separately; MAPPO pairs a team-level advantage with per-agent ratios; HAPPO employs sequential ratios with per-agent advantages; and single-agent reductions operating on factorized joint policies aggregate both into fully joint products.
We formalize these two design choices as support matrices $\SA$ and $\SR$, and prove a canonical structure: the expected multi-agent policy optimization objective depends on the pair $(\SA,\SR)$ only through their matrix product $\tS=\SR\SA$. This yields two key consequences: \emph{(i) Redundancy:} the two support matrices are interchangeable with respect to the signal, meaning neither aggregation pattern is inherently superior. \emph{(ii) Variance Ordering:} the advantage aggregates rewards as a sum (additive variance with an interior bias-variance optimum at the coupling neighborhood), whereas the ratio aggregates likelihood ratios as a product (multiplicative variance that grows exponentially with support size, with no accompanying bias reduction). The resulting design principle is unambiguous: aggregate neighbors in the advantage, sized to the coupling neighborhood, and keep the ratio per-agent. This explains why neighbor-based advantages are more prevalent than neighbor-based ratios in prior heuristic and empirical designs.
We prove these results under specified assumptions and validate them across four carefully designed synthetic cooperative games and a real-world large-scale traffic-signal control task. The code for our experiments is available at \url{https://github.com/RS2002/MAPO}.
\end{abstract}

\section{Introduction}

Cooperative Multi-Agent Reinforcement Learning (MARL) built on policy optimization, represented by Proximal Policy Optimization (PPO) \citep{schulman2017proximal} and Trust Region Policy Optimization (TRPO) \citep{schulman2015trust}, faces a recurring question when updating each agent: \emph{How many other agents should be aggregated into that agent's update?} This question is easy to overlook because it is answered twice, in two different places:
\begin{itemize}[leftmargin=1.4em,itemsep=1pt,topsep=2pt]
\item \textbf{Advantage support} $\SA$: which agents' rewards are summed into agent $i$'s advantage---the credit signal that indicates how good the joint action was for $i$; and
\item \textbf{Ratio support} $\SR$: which agents' likelihood ratios are multiplied into agent $i$'s clipped importance weight---the trust-region object that keeps the update near the behavior policy.
\end{itemize}
Both range from per-agent (aggregate no one else) to joint (aggregate everyone), and existing methods sit at different points on each: IPPO \citep{de2020independent} keeps both per-agent, while MAPPO \citep{yu2022mappo} pairs a team advantage (all rewards) with a per-agent ratio; Traffic and networked methods use neighborhood advantages \citep{chu2020ma2c} with per-agent ratios; Sequential learning methods employ joint sequential ratio products while keeping the advantage per-agent \citep{kuba2022trust, zhong2024heterogeneous}; Casting the whole team as a single multi-action agent and running vanilla PPO on the factorized joint policy (ratio) and joint advantage yields a route taken by centralized single-agent reductions of MARL \citep{cmat2025, zhao2026triplebert}. Yet the two supports have never been analyzed together, and there is no principle for where to sit on either, which this paper aims to address through canonical-form analysis.

Our starting point is that the two supports are not independent degrees of freedom. We prove (Theorem~\ref{thm:canon}) that the expected multi-agent policy gradient depends on the pair $(\SA,\SR)$ only through their matrix product $\tS=\SR\SA$. This is a gauge freedom: $\SA$ and $\SR$ are interchangeable up to their product, so neither support is intrinsically redundant; any cross-agent aggregation placed in the ratio can equally be placed in the advantage, and vice versa (Corollary~\ref{cor:canon}). What breaks this symmetry is variance, not signal. Aggregating rewards through the advantage is a sum of bounded terms (additive variance), whereas aggregating ratios is a product of importance weights, whose variance compounds multiplicatively in the number of factors (Lemma~\ref{lem:var}). Since the two routes realize the same expected gradient but the ratio route strictly inflates variance, the variance-optimal choice is to place all cross-agent aggregation in the advantage and keep the ratio per-agent (Corollary~\ref{cor:design}). It is in this precise, variance-ordered sense that cross-agent importance ratios are redundant: they add no signal that the advantage cannot, and cost variance that the advantage does not. This also explains a standing asymmetry in previous empirical and heuristic designs: neighborhood advantages are widely used, while local ratios are comparatively rare.

In conclusion, we contribute: (i) a canonical form showing that the two supports enter the expected gradient only through $\tS=\SR\SA$, establishing their mutual redundancy; (ii) a variance-ordering that breaks the tie---advantage sums are additive while ratio products are exponential---yielding the design rule: aggregate in the advantage and keep the ratio per-agent; and (iii) an empirical validation across synthetic games and a real-world traffic signal control task that pin down when each effect appears. We discuss connections to prior work throughout, and provide an extended treatment in Appendix~\ref{sec:related}.

\section{Preliminaries}\label{sec:prelim}

\paragraph{Single-agent policy optimization.}
PPO \citep{schulman2017proximal} optimizes a policy $\pi_{\bm\theta}$ by taking several gradient steps on a batch collected under a behavior policy $\pi_{\bm\theta_{\mathrm{old}}}$, correcting the off-policy mismatch with a clipped importance ratio $\varrho(\bm\theta)=\pi_{\bm\theta}(a\mid s)/\pi_{\bm\theta_{\mathrm{old}}}(a\mid s)$ and the advantage $A$, estimated by Generalized Advantage Estimation (GAE) \citep{schulman2015high}:
\begin{equation}
L^{\mathrm{PPO}}(\bm\theta)=\mathbb E\big[\min\big(\varrho A,\ \mathrm{clip}(\varrho,
1{-}\epsilon,1{+}\epsilon)A\big)\big].
\label{eq:ppo}
\end{equation}
Two objects drive the update: the advantage $A$ (the credit signal) and the ratio $\varrho$ (the trust-region weight).

\paragraph{Cooperative multi-agent policy optimization.}
In a cooperative Markov game \citep{littman1994markov} (formalized in Appendix \ref{app:mamdp}) with $n$ agents, each agent $i$ has a policy $\pi^i_{\bm\theta}$; under Centralized-Training with Decentralized-Execution (CTDE) and Centralized-Training with Centralized-Execution (CTCE) schemes alike \citep{jin2025comprehensive}, the joint policy factorizes as $\pi_{\bm\theta}(\bm a\mid s)=\prod_i\pi^i_{\bm\theta}(a^i\mid s)$. Lifting \eqref{eq:ppo} to $n$ agents requires two design decisions that are usually made implicitly:
\begin{enumerate}[leftmargin=1.4em,itemsep=1pt,topsep=2pt]
\item \textbf{Which reward drives agent $i$'s advantage?} MAPPO \citep{yu2022mappo} uses the team advantage (all agents share one $A$ built from the global reward); independent learners (IPPO) \citep{de2020independent} use each agent's own reward; value-decomposition, difference-reward, and neighbor-advantage methods sit in between.
\item \textbf{Whose ratios enter agent $i$'s clipped weight?} Independent and decentralized actors (IPPO, MAPPO) clip a per-agent ratio $\varrho^i$. A fully centralized controller that treats $\pi=\prod_i\pi^i$ as one policy clips the joint ratio $\prod_j\varrho^j$, in which every factor is differentiated jointly \citep{cmat2025}. Sequential methods such as HAPPO \citep{kuba2022trust} form a compound ratio $\prod_{j\le i}\varrho^j$. We study the jointly-differentiated case, whose supports range from per-agent ($\SR=I$) to fully joint ($\SR=\one\one^\top$).
\end{enumerate}
These are the two knobs this paper isolates. We call the first the advantage support and the second the ratio support, and formalize both as $0/1$ matrices in Section~\ref{sec:setup}. 

\paragraph{Why the choice matters: a bias--variance tension.} Both supports can reduce the same bias, but they pay for it very differently---a tension we preview here with a direct measurement of the multi-agent policy gradient, deferring the full study to Sec.~\ref{sec:exp}. On a single-step dense pairwise game ($n=20$ agents, $K=4$ actions each, ring coupling of radius $\rho^\star=4$; full definition in Sec.~\ref{sec:exp:setup}), we fix a reference policy, estimate the true team-return policy gradient by Monte Carlo, and then form the clipped-surrogate gradient estimator induced by a given pair of supports. For that estimator we report three quantities (all of the gradient estimator, in log scale): its squared bias, variance, and their sum, i.e., the Mean Square Error (MSE). Here a support of radius $\rho$ means each agent aggregates its $\rho$-hop neighborhood on the coupling graph: $\rho=0$ is per-agent (aggregate no one else), $\rho=\rho^\star$ exactly covers the true coupling neighborhood, and larger $\rho$ over-aggregates. Panels (a,b) vary one support at a time---the advantage radius $\rho_A$ with the ratio held per-agent (a), and the ratio radius $\rho_R$ with the advantage held per-agent (b). The results show that the two bias curves are identical: enlarging either support removes the same missing-coupling bias, so on the expected gradient the two supports are interchangeable (a redundancy we prove in Sec.~\ref{sec:setup}). The two variance curves, in contrast, differ sharply: the advantage aggregates rewards additively, so its variance grows gently and its MSE bottoms out near the coupling radius $\rho^\star$, whereas the ratio aggregates likelihood ratios multiplicatively, so its variance grows far faster and its MSE turns up at a smaller radius. Panels (c,d) isolate this: they fix a matched effective support (so both realize the same expected gradient) and compare the two ways of reaching it---aggregating through the advantage (path $\mathbf P$) versus through the ratio (path $\mathbf Q$)---plotting their bias (c) and variance (d) as the support grows; the bias tracks together while the variance separates by up to $9\times$. Same benefit, different cost---understanding why, and what it implies for where to sit on each axis, is the question we take up.
\begin{figure}[t!]
\centering
\includegraphics[width=\textwidth]{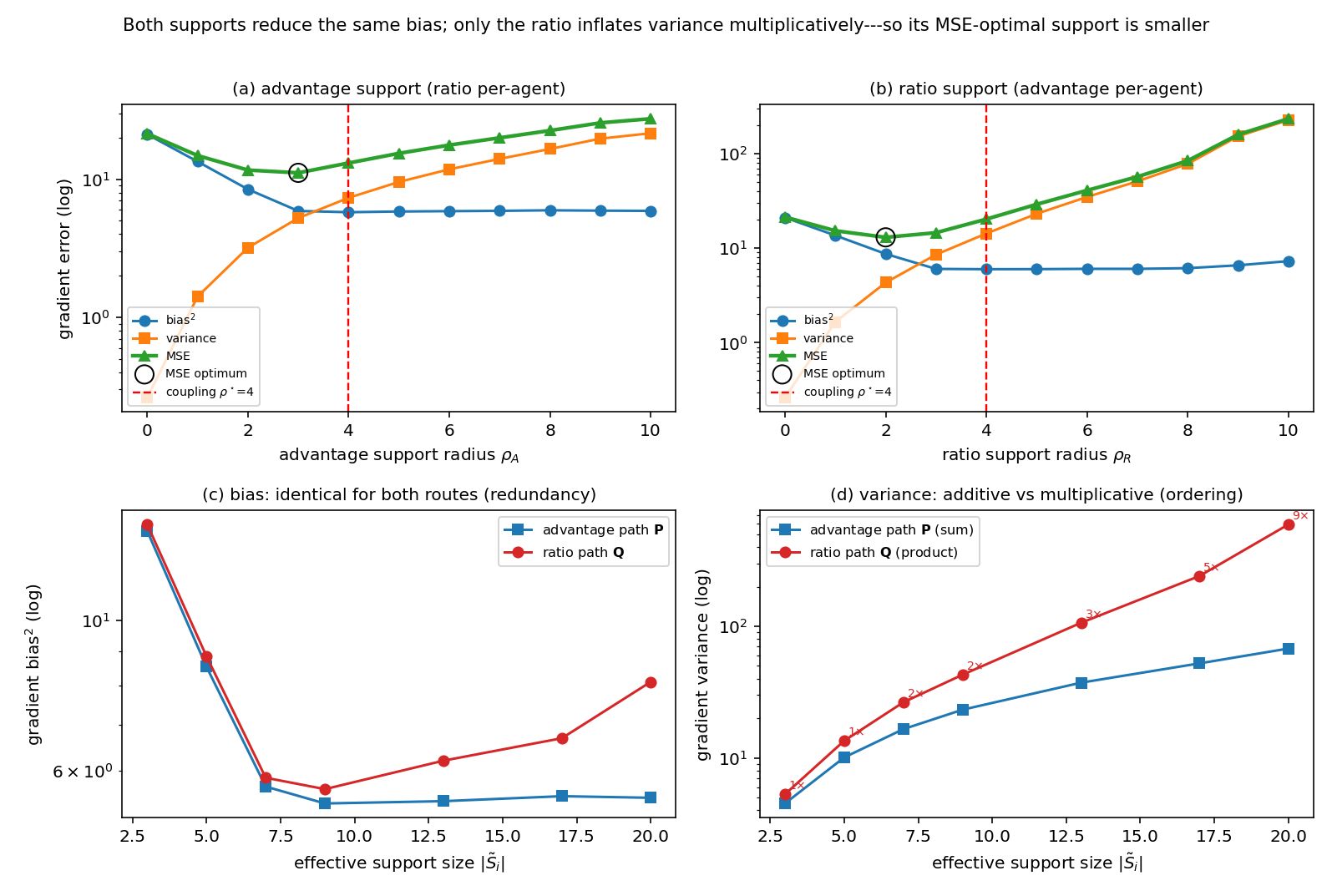}
\caption{Bias, variance, and MSE of the multi-agent policy-gradient estimator on a single-step dense-pairwise game, as a function of the aggregation support (log axes throughout). (a) Sweep of the advantage-support radius $\rho_A$ (with a per-agent ratio) and (b) sweep of the ratio-support radius $\rho_R$ (with a per-agent advantage): each panel plots the estimator's squared bias, its variance, and their sum (MSE), with the MSE minimum marked ($\circ$) and the coupling radius $\rho^\star=4$ shown dashed. (c,d) At a matched effective support $\tS=\text{ring}$ of increasing size, the advantage path $\mathbf P$ ($\SA=\text{ring}, \SR=I$) and the ratio path $\mathbf Q$ ($\SA=I, \SR=\text{ring}$), which share the same expected gradient: (c) squared bias and (d) variance of each path, the annotations in (d) giving the ratio $\Var(\mathbf Q)/\Var(\mathbf P)$.}
\label{fig:motiv}
\end{figure}

\section{Analysis: Two Supports, One Product}\label{sec:analysis}

\subsection{Problem Setup}\label{sec:setup}

We first analyze a single decision step; Remark~\ref{rem:mdp} and the episodic experiment (Sec.~\ref{sec:exp}) provide the finite-horizon extension.

\paragraph{Agents, policies, actions.}
Let $\mathcal N=\{1,\dots,n\}$ be the set of agents and $\mathcal A=\{1,\dots,K\}$ a finite action set. Agent $i$ has policy $\pi^i_{\bm\theta}(\cdot)\in\Delta(\mathcal A)$ with its own parameter block $\bm\theta^i$; the full parameter is $\bm\theta=(\bm\theta^1,\dots,\bm\theta^n)$ with disjoint blocks. The joint policy factorizes as
\begin{equation}
\pi_{\bm\theta}(\bm a)=\prod_{i=1}^n\pi^i_{\bm\theta}(a^i),\qquad
\bm a=(a^1,\dots,a^n)\in\mathcal A^n .
\label{eq:factor}
\end{equation}

\paragraph{Rewards and coupling.}
Agent $i$ receives reward $r_i(\bm a)$; the team return is $\R(\bm a)=\sum_i r_i(\bm a)$. The coupling graph has symmetric adjacency $C\in\{0,1\}^{n\times n}$ with $C_{ii}=1$, where $C_{ij}=1$ iff $r_i$ depends on $a^j$. We write $\partial i=\{j:C_{ij}=1\}$. (For simplicity, we omit the dependence on the state in the reward notation here.)

\paragraph{Two supports.}
Fix a behavior parameter $\bm\theta_{\mathrm{old}}$ and, for a candidate $\bm\theta$, define the per-agent likelihood ratio $\varrho_j(\bm a)=\pi^j_{\bm\theta}(a^j)/\pi^j_{\bm\theta_{\mathrm{old}}}(a^j)$. Let $\SA,\SR\in\{0,1\}^{n\times n}$ be symmetric support matrices with $S_{ii}=1$. The advantage and importance weight of agent $i$ are given by
\begin{equation}
A_i=\sum_{j:\,\SA_{ij}=1}\,r_j(\bm a)-b_i,
\qquad
w_i(\bm\theta)=\!\!\prod_{j:\,\SR_{ij}=1}\!\!\varrho_j(\bm a),
\label{eq:supports}
\end{equation}
with $b_i$ any control variate independent of $\bm a$. Special cases include $\SA=I$ (independent advantage), $\SA=\one\one^\top$ (team advantage), $\SR=I$ (per-agent ratio), and $\SR=\one\one^\top$ (joint ratio $\prod_j\varrho_j$).

\paragraph{Surrogate and gradient.}
The (unclipped) surrogate is $L(\bm\theta)=\sum_i\Eold[w_i(\bm\theta)A_i]$, with $A_i$ detached (constant in $\bm\theta$). Since $\nabla_{\bm\theta}w_i=w_i\sum_{j:\SR_{ij}=1}\nabla_{\bm\theta}\log\pi^j(a^j)$ and $\nabla_{\bm\theta^m}\log\pi^j=0$ for $j\ne m$, we have
\begin{equation}
\bm g_m(\bm\theta)=\nabla_{\bm\theta^m}L
=\Eold\Big[\Big(\textstyle\sum_{i:\,\SR_{im}=1}w_i(\bm\theta)A_i\Big)
\nabla_{\bm\theta^m}\log\pi^m(a^m)\Big].
\label{eq:grad}
\end{equation}

\subsection{Assumptions and Their Justification}\label{sec:assump}

\begin{assumption}[Factorized policy / conditional action independence]
\label{as:indep}
Under both $\bm\theta_{\mathrm{old}}$ and any candidate $\bm\theta$, the joint policy factorizes as in \eqref{eq:factor}; equivalently, given the parameters, actions are sampled independently across agents.
\end{assumption}

\noindent\textit{Justification.} This is not an extra modeling restriction but rather the defining structure of the methods under study. In most MARL algorithms (e.g. IPPO, MAPPO, HAPPO) each agent samples $a^i\sim\pi^i(\cdot\mid o^i)$ independently by construction, so \eqref{eq:factor} holds exactly. The assumption would fail only for policies with an explicitly autoregressive action head (e.g., a decoder that conditions $a^i$ on realized $a^{<i}$); we exclude that case and note that it corresponds to a different (chain-rule) factorization \citep{wen2022multi}. Importantly, Assumption~\ref{as:indep} concerns conditional independence given parameters (and context); it does not claim that the agents' actions are marginally uncorrelated---they are correlated through shared observations and, during learning, through the coupled reward.

\begin{assumption}[Full support / finite second moment]\label{as:supp}
$\pi^i_{\bm\theta}(a)>0$ for all $i,a$ and all $\bm\theta$ in a neighborhood of $\bm\theta_{\mathrm{old}}$, so each $\varrho_j$ is finite with $\mathbb E[\varrho_j^2]<\infty$.
\end{assumption}

\noindent\textit{Justification.} Softmax policies, the standard parameterization, satisfy full support automatically. Finite second moments are required for any importance-weighted estimator to have finite variance and are standard in PPO and TRPO analysis.

\begin{assumption}[Symmetric coupling]\label{as:sym}
$C$ is symmetric with $C_{ii}=1$, and $r_i$ depends on $a^j$ iff $C_{ij}=1$.
\end{assumption}

\noindent\textit{Justification.} Symmetry holds whenever coupling arises from shared or pairwise terms (shared resources, pairwise congestion or collision, common team reward), which covers the cooperative tasks of interest. The directional case is a routine extension (Remark~\ref{rem:dir}) obtained by replacing $C$ with the directed influence matrix; the main theorem does not require symmetry, which we adopt only to state the coupling radius cleanly.

\begin{assumption}[Bounded rewards]\label{as:bdd}
$|r_i(\bm a)|\le r_{\max}<\infty$ for all $i,\bm a$.
\end{assumption}

\noindent\textit{Justification.} This is standard and holds for any bounded reward; by truncation, it also applies to sub-Gaussian rewards up to negligible tails. It ensures that the advantage-side aggregation has variance $O(|\SA|)$ rather than exponential.

\subsection{Ratio Moments}
\begin{lemma}[Unbiased weight]\label{lem:mean}
Under Assumptions~\ref{as:indep}--\ref{as:supp}, for any $S\subseteq\mathcal N$, $\Eold\big[\prod_{j\in S}\varrho_j\big]=1$; in particular $\mathbb E[w_i]=1$ for every ratio support.
\end{lemma}
This follows because the behavior-policy expectation factorizes over agents and each per-agent ratio has unit mean; a proof is provided in Appendix~\ref{app:proof:mean}.

\begin{lemma}[Multiplicative variance]\label{lem:var}
Under Assumptions~\ref{as:indep}--\ref{as:supp}, with $\chi^2_j:=\chi^2(\pi^j_{\bm\theta}\,\|\,\pi^j_{\mathrm{old}})=\mathbb E[\varrho_j^2]-1 \ge 0$,
\begin{equation}\label{eq:varprod}
\Var\Big(\prod_{j\in S}\varrho_j\Big)=\prod_{j\in S}(1+\chi^2_j)-1 .
\end{equation}
Hence the weight variance is nondecreasing in $S$, and if $\chi^2_j\ge c>0$ then it is $\ge(1+c)^{|S|}-1$, i.e., exponential in $|S|$.
\end{lemma}

The multiplicative form follows from the independence of the per-agent ratios, which causes the second moments to factorize; the full derivation is given in Appendix~\ref{app:proof:var}.
\begin{remark}
At $\bm\theta=\bm\theta_{\mathrm{old}}$, every $\chi^2_j=0$ and $w_i\equiv1$, so the ratio only activates once $\bm\theta$ departs from $\bm\theta_{\mathrm{old}}$, i.e., across PPO's inner epochs. This is precisely why the two supports coincide on-policy (Thm.~\ref{thm:canon}) yet diverge off-policy (Prop.~\ref{prop:dom}).
\end{remark}

\subsection{Main results}\label{sec:mainresults}

\subsubsection{The expected gradient factorizes through $\tS=\SR\SA$}

\begin{theorem}[Support factorization]\label{thm:canon}
Under Assumptions~\ref{as:indep}--\ref{as:sym}, at $\bm\theta=\bm\theta_{\mathrm{old}}$
the expected gradient \eqref{eq:grad} is
\begin{equation}
\bm g_m=\sum_{j:\,C_{jm}=1}\tS_{mj}\,\mathbb E\big[r_j\,\nabla_{\bm\theta^m}
\log\pi^m(a^m)\big],\qquad \tS=\SR\SA ,
\label{eq:tildeS}
\end{equation}
where $\tS$ is the matrix product of $\SR$ and $\SA$. Thus $\bm g_m$
depends on $(\SA,\SR)$ only through $\tS$: any two support pairs with the same
product $\SR\SA$ induce the same expected gradient.
\end{theorem}
The key step evaluates the surrogate gradient at $\bm\theta_{\mathrm{old}}$, where $w_i{=}1$, and collects terms by the score-function identity; the complete proof is in Appendix~\ref{app:proof:canon}.

\subsubsection{Canonical form and redundancy of cross-agent ratios}

\begin{corollary}[Per-agent ratio canonical form]\label{cor:canon}
Under Theorem~\ref{thm:canon}, every estimator with supports
$(\SA,\SR)$ has, at the on-policy point, the same expected gradient as the
\emph{per-agent-ratio} estimator $(\SR{=}I,\ \SA{=}\tS)$ whose advantage is the
reweighted quantity $A'_m=\sum_j\tS_{mj}r_j$ with $\tS=\SR\SA$. In particular the
joint ratio $\SR=\one\one^\top$ paired with advantage $\SA$ is gradient-equivalent
to a per-agent ratio with advantage reweighted by $\one\one^\top\SA$. Hence
\emph{cross-agent importance ratios add no expected-gradient signal beyond a
linear reweighting of the advantage.}
\end{corollary}
This is immediate from Theorem~\ref{thm:canon} by reading off the coefficient of each score term (Appendix~\ref{app:proof:corcanon}).

\begin{remark}[The team-advantage scalar]\label{rem:scalar}
With team advantage $\SA=\one\one^\top$, per-agent ratio gives $\tS=\one\one^\top$
while joint ratio gives $\tS=\one\one^\top\one\one^\top=n\,\one\one^\top$. The two
expected gradients are therefore \emph{identical up to the global scalar $n$}:
the joint ratio merely rescales the effective step by the support size. Any fair
comparison must control for this scalar (we do so in Sec.~\ref{sec:exp} by
matching effective step size).
\end{remark}

\subsubsection{Variance Domination}

We now compare the two canonical realizations of a target product $\tS=C$:
\begin{equation}\label{eq:pq}
\text{(P) advantage path: }\SA=C,\ \SR=I;\qquad
\text{(Q) ratio path: }\SA=I,\ \SR=C .
\end{equation}
Both realize the same expected gradient $\sum_{j\in\partial m}\mathbb E[r_j\nabla\log\pi^m]$ by Theorem~\ref{thm:canon}.

\begin{proposition}[Off-policy variance domination]\label{prop:dom}
Under Assumptions~\ref{as:indep}--\ref{as:bdd}: (i) at $\bm\theta=\bm\theta_{\mathrm{old}}$, P and Q have identical mean and identical variance; (ii) for $\bm\theta\ne\bm\theta_{\mathrm{old}}$, the P-weight $\sum_{j\in\partial m}r_j$ is bounded by $|\partial m|\,r_{\max}$ with variance independent of $\bm\theta$, whereas the Q-weight $\sum_{i\in\partial m}\bigl(\prod_{j\in\partial i}\varrho_j\bigr)r_i$ has second moment $\mathbb E[(\prod_{j\in\partial i}\varrho_j)^2]=\prod_{j\in\partial i}(1+\chi^2_j)$, so once any $\chi^2_j>0$ its variance is bounded below by a positive multiple of $\prod_{j\in\partial m}(1+\chi^2_j)-1$. Hence $\Var(\bm g^Q_m)/\Var(\bm g^P_m)\to\infty$ as the policy step or $|\partial m|$ grows.
\end{proposition}
The bound follows by comparing the bounded P-weight to the Q-weight's exploding second moment from Lemma~\ref{lem:var}; see Appendix~\ref{app:proof:dom}.

\begin{remark}[Clipping widens the gap]\label{rem:clip}
Reinstating the PPO clip applies a per-weight, $1$-Lipschitz truncation. For P, the clipped object is a single low-variance ratio $\varrho_m$, so clipping rarely activates; for Q, it is the heavy-tailed product (Lemma~\ref{lem:var}), which triggers clipping more often, injecting truncation bias and leaving the retained variance elevated. Clipping thus cannot reverse Proposition~\ref{prop:dom} and typically amplifies it.
\end{remark}

\begin{corollary}[Design rule]\label{cor:design}
Among all $(\SA,\SR)$ that realize a target unbiased product $\tS$, the choice $\SR=I$ (per-agent ratio) with $\SA=\tS$ minimizes the gradient-estimator variance. In particular, the joint (compound) ratio is weakly dominated on-policy and strictly dominated off-policy: it should be replaced by a per-agent ratio with the corresponding advantage reweighting.
\end{corollary}

\subsubsection{The Advantage Support: A Bias--Variance Tradeoff}\label{sec:advtradeoff}


At a matched effective support, the ratio support does not provide any bias reduction advantage over the advantage support: any bias it can remove is already removed by the advantage support at strictly lower variance (Corollary~\ref{cor:canon}, Prop.~\ref{prop:dom}). Consequently, its variance-optimal setting is the per-agent boundary $\SR=I$. The advantage support is qualitatively different: with the ratio fixed at per-agent, it trades bias against variance solely through $\SA$, and this tradeoff has an interior optimum. This asymmetry between the two supports is fundamental, and we can characterize the advantage side as sharply as the ratio side.

\begin{proposition}[Advantage-support bias--variance tradeoff]\label{prop:adv}
Fix the per-agent ratio $\SR=I$ and consider the on-policy estimator $\bm g_m$ with advantage support $\SA$. Write the true (team-return) gradient block as $\bm g_m^\star=\sum_{j\in\partial m}\mathbb E[r_j\nabla_{\bm\theta^m}\log\pi^m]$. Then, under Assumptions~\ref{as:indep}--\ref{as:bdd}:
\begin{enumerate}[label=(\roman*),leftmargin=1.9em,itemsep=1pt,topsep=2pt]
\item (Bias) $\mathbb E[\bm g_m]=\sum_{j\in\partial m}\SA_{mj}\, \mathbb E[r_j\nabla_{\bm\theta^m}\log\pi^m]$, so $\bm g_m$ is unbiased iff $\SA_{mj}=1$ for every coupled agent $j\in\partial m$; a support that misses a coupled agent omits its term and is biased.
\item (Variance) Including an uncoupled agent $j\notin\partial m$ leaves the mean unchanged (its term has zero expectation by the score-function identity) but adds a nonnegative variance contribution, strictly positive whenever $r_j$ is conditionally non-degenerate given $a^m$.
\end{enumerate}
Consequently, the mean-squared-error-optimal advantage support is exactly the coupling neighborhood $\partial m$: smaller supports are biased, larger supports inflate variance.
\end{proposition}
The bias and variance terms are computed separately and their sum minimized; the derivation is in Appendix~\ref{app:proof:adv}.


\begin{remark}[The asymmetry, precisely]\label{rem:asym}
Propositions~\ref{prop:dom} and~\ref{prop:adv} together establish the central asymmetry of the paper as a theoretical result, not merely an empirical observation. With the ratio fixed to per-agent, the advantage support has a bias term that larger supports remove, yielding an interior MSE optimum at the coupling neighborhood. In contrast, the ratio support offers no bias-reduction benefit: by the canonical form (Corollary~\ref{cor:canon}), any bias reduction it could achieve is equally attainable through the advantage, and by Proposition~\ref{prop:dom}, at strictly lower variance. Hence, on the variance-optimal frontier, the ratio is never used to reduce bias, and its optimum is the boundary $\SR=I$. Although the two knobs appear symmetric in the surrogate \eqref{eq:supports}, they play categorically different roles.
\end{remark}

\begin{remark}[Directional-coupling extension]\label{rem:dir}
If coupling is directional, replace $C$ by the directed influence matrix $D$ with $D_{ij}=1$ iff $r_i$ depends on $a^j$. Theorem~\ref{thm:canon} holds with the survival condition $D_{jm}=1$ and $\tS=\SR\SA$; the variance-optimal design uses $\SR=I$, $\SA=D^\top$.
\end{remark}

\begin{remark}[Finite-horizon extension]\label{rem:mdp}
In a finite-horizon Markov game, the argument applies per time step to the GAE advantage and the per-step ratio. The multiplicative variance of Lemma~\ref{lem:var} then compounds across agents and time, so the joint ratio's variance grows in $n\cdot H$; the on-policy factorization of Theorem~\ref{thm:canon} is unchanged. The episodic experiment (Sec.~\ref{sec:exp}) confirms the redundancy in this setting.
\end{remark}

\section{Experiments}\label{sec:exp}

\subsection{Experiment Setup}\label{sec:exp:setup}
The theory yields three testable predictions. \textbf{(R) Redundancy:} estimators with the same product $\tS=\SR\SA$ have identical expected gradients. \textbf{(D) Domination:} realizing a fixed $\tS$ through the ratio product incurs strictly higher variance off-policy than through the advantage sum, with the gap growing multiplicatively in the support size. (Fig. \ref{fig:motiv}) \textbf{(A) Asymmetry:} at the gradient-estimator level, the advantage support presents a bias--variance tradeoff with an interior MSE optimum at the coupling neighborhood (Prop.~\ref{prop:adv}), whereas the ratio support has no bias effect and a boundary optimum at $\rho_R=0$ (Prop.~\ref{prop:dom}). During training, the advantage radius influences the learned policy only insofar as bias alters the equilibrium---decisive when the externality is not self-internalized, mild otherwise---while the ratio radius, being bias-free, never changes the learned policy and affects only stability. We test these predictions through a multi-family training study that traces how each support affects the learned policy across four coupling structures and a scaling study that identifies the regime in which the ratio's variance penalty becomes manifest.

All environments are cooperative games with a prescribed coupling graph $C$, so the ground-truth coupling neighborhood is known exactly. Each agent $i$ observes a phase or context and selects $a^i\in\{1,\dots,K\}$; per-agent rewards $r_i$ are defined in Appendix~\ref{app:setup}, and the team return is $\R=\sum_i r_i$; ring neighborhoods of radius $\rho^\star$ define $\partial i$. We use four game families that span the regimes distinguished by the theory: a dense pairwise game (symmetric payoff over a ring neighborhood, with internalized externality); a directed dilemma, where a high-value action imposes a nuisance on downstream successors that the actor does not bear (a non-internalized externality); local congestion, where neighbors choosing the same resource split its value (crowding partly self-borne); and a block community game with all-to-all coupling inside disjoint blocks (a non-geometric graph). A more detailed description is provided in Appendix~\ref{app:setup}. In addition, we provide a real-world large-scale traffic signal control validation in Appendix~\ref{app:traffic:ratio}.

\subsection{Experiment Results}

We first examine PPO training across the four coupling families, which differ in how agents interact. Each is a finite-horizon Markov game with phase-conditioned tabular actors trained by PPO (clip $0.2$); we sweep the advantage radius $\rho_A$ (with per-agent ratio fixed) and, separately, the ratio radius $\rho_R$ (with team advantage fixed). Curves are converged over the final $40$ iterations and averaged over $5$ seeds.

\textbf{The advantage radius selects the outcome only for non-internalized externalities (Fig.~\ref{fig:advsweep}).} Figure~\ref{fig:advsweep} plots, for each of the four coupling families, the final team return as a function of the advantage radius $\rho_A$ (per-agent ratio fixed), with the true coupling radius $\rho^\star$ marked. The four families separate sharply. In the directed dilemma, the advantage radius is decisive: independent advantage ($\rho_A=0$) converges to the tragedy equilibrium (team return $-159$), while any $\rho_A \ge 1$ internalizes enough of the successor externality to reach the social optimum ($+158$). In the other three families, the effect is mild and the outcome-optimal radius is small ($\rho_A \in \{0,1\}$): when an agent already bears (part of) the coupling it induces—symmetric pairwise payoffs, shared congestion, within-block rewards—the independent advantage is nearly unbiased, so enlarging $\rho_A$ buys little signal and adds variance, and the return is flat or slightly decreasing in $\rho_A$. This is the training-level counterpart of the estimator tradeoff in Fig.~\ref{fig:motiv}a: the advantage radius should match the coupling neighborhood, but the outcome moves only when the un-internalized part of the externality is large enough to change the equilibrium. Full training curves for all four families are provided in Appendix~\ref{app:figs}.

\begin{figure}[t]
\centering
\includegraphics[width=\textwidth]{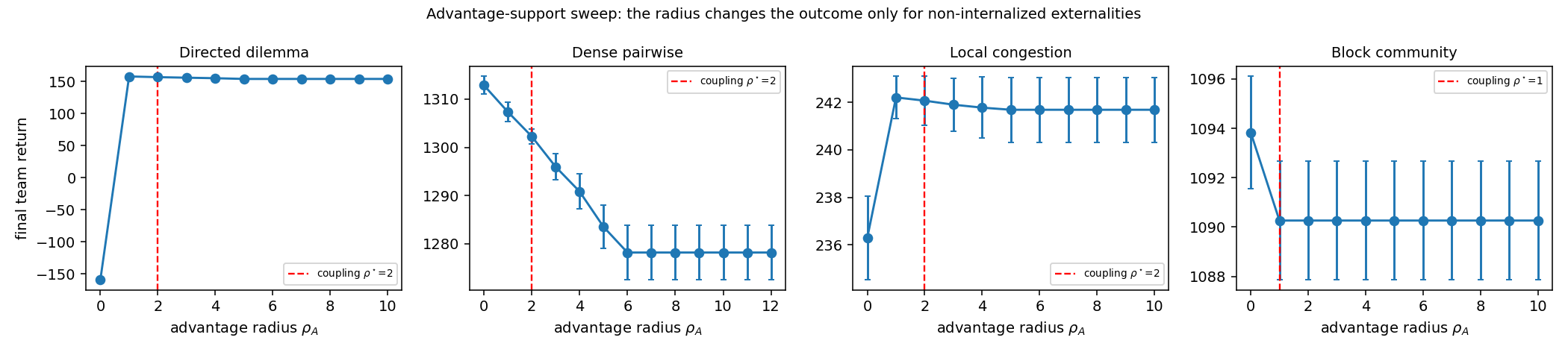}
\caption{Advantage-support sweep across the four coupling families: final team return versus the advantage radius $\rho_A$, with the ratio held per-agent. One panel per family; the dashed vertical line marks the coupling radius $\rho^\star$. Bands indicate $\pm1$ standard deviation across seeds; in the directed dilemma, the standard deviation is below $0.4$ and smaller than the marker, as the two equilibria are reached near-deterministically.}
\label{fig:advsweep}
\end{figure}

Moreover, redundancy implies that the ratio radius carries no benefit; Proposition~\ref{prop:dom} states that it carries a variance cost that grows multiplicatively with support size and with the per-update policy shift $\chi^2$. Whether that cost is visible in training therefore depends on how far off-policy each update travels—which a trust-region method controls through its step size. We make this explicit by scaling the number of agents under two configurations that differ only in aggressiveness: a standard PPO setting ($\mathrm{lr}=3\times10^{-3}$, batch $64$) and an off-policy stress setting ($\mathrm{lr}=0.15$, batch $16$); both use the team advantage and match the effective step across ratio supports (Remark~\ref{rem:scalar}), so any difference is due to variance, not step size.

\textbf{Masked under conservative updates, revealed under stress (Fig.~\ref{fig:scaling}).} Figure~\ref{fig:scaling} varies the number of agents $n$ (log axis) for each family and reports two quantities: the joint/per-agent final return ratio (solid, left axis)—indicating how much return the joint ratio sacrifices relative to a per-agent ratio—and the joint ratio's clip fraction (dashed, right axis), under both the standard (blue) and off-policy stress (red) configurations. Under the standard configuration, the joint and per-agent ratios are statistically indistinguishable up to $n=160$ (return ratio $1.00$, joint clip fraction $\approx 0$), with a $4$–$8\%$ gap emerging only at $n=320$: the trust region keeps $\chi^2$ so small ($\chi^2_j \approx 6\times10^{-4}$ per update) that the variance factor $(1+\chi^2)^n-1$ is negligible even for $n$ in the hundreds. Under the stress configuration, the same experiment reproduces the folklore instability across all four families: the joint ratio's clip fraction rises monotonically with $n$ (reaching $0.2$–$0.45$ at the largest $n$ tested per family), while the per-agent ratio's clip fraction remains at $0$ throughout—the signature of Lemma~\ref{lem:var}. A product of $n$ ratios leaves the trust region almost surely once $n$ is large, so it is clipped increasingly often and its gradient is throttled, whereas a single ratio never is. Where the coupling is sharp enough that this throttling starves the update of signal—the directed dilemma and dense pairwise families—the return degrades with it. Where the reward is more forgiving (local congestion, block community), the same clipping leaves the final return closer to parity, making the clip fraction the more universal diagnostic.

\begin{figure}[t]
\centering
\includegraphics[width=\textwidth]{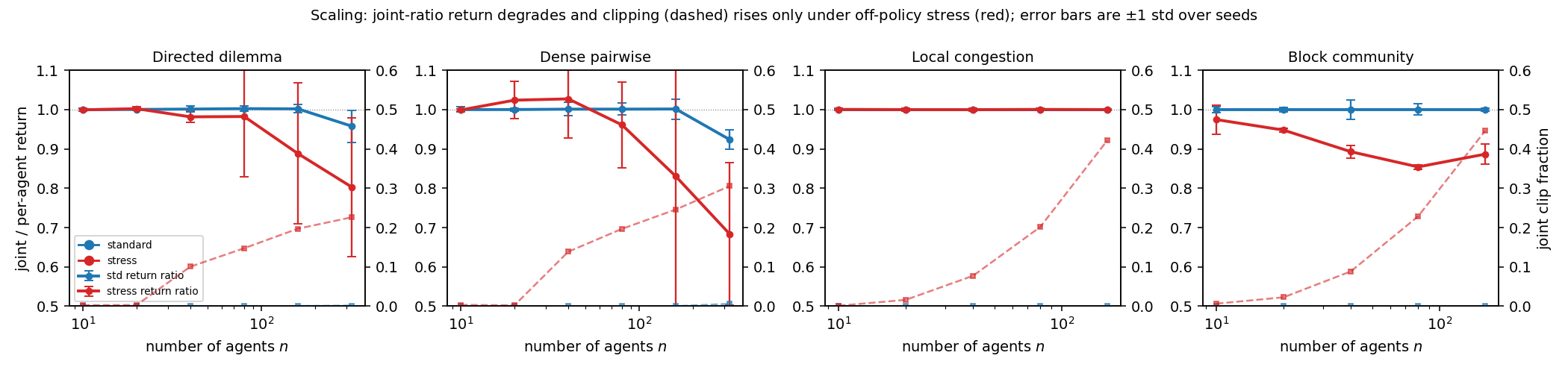}
\caption{Agent-count scaling across the four coupling families. For each family and agent count $n$ (log axis), solid curves report the joint/per-agent final return ratio (left axis), and dashed curves show the joint ratio's clip fraction (right axis). Results are shown under a standard PPO configuration (blue; $\mathrm{lr}=3\times10^{-3}$, batch $64$) and an off-policy stress configuration (red; $\mathrm{lr}=0.15$, batch $16$); both use the team advantage and match the effective step across ratio supports. Error bars indicate $\pm1$ standard deviation over seeds (on the standard curves they are smaller than the markers).}
\label{fig:scaling}
\end{figure}

\begin{remark}[Step size, not agent count, is the trigger]\label{rem:scale}
The variance factor $(1+\chi^2)^n-1$ has two levers: the support size $n$ and the per-update shift $\chi^2$. Fig.~\ref{fig:scaling} shows they act together—the gap grows with $n$, but only once $\chi^2$ is non-negligible. A conservative learning rate makes $\chi^2$ vanishingly small, so even $n$ in the hundreds is harmless; an aggressive rate makes it bite. This reconciles the theorem with the common observation that MAPPO's per-agent ratio and various joint-ratio schemes often perform comparably: they do, precisely in the near-on-policy regime where the penalty is dormant. The per-agent ratio weakly dominates always and strictly dominates whenever updates are pushed off-policy—at no cost, since the benefit of a larger ratio support is exactly zero.
\end{remark}

\begin{remark}[Why we recommend against the joint ratio, not merely note its risk]\label{rem:norec}
It might seem that, since the joint ratio is harmless in the near-on-policy regime, one could simply use it with a small step. We caution against this for a practical reason that our controlled environments understate. Here we know the coupling, the reward scale, and the effective off-policyness, so we can see that the penalty is dormant. In a real large-scale system, one does not know these in advance: the effective $\chi^2$ varies across states, agents, and training phases, and a minibatch that happens to be more off-policy—after a reward spike, an exploration burst, or a learning-rate warmup—can push the $n$-fold product out of the trust region and stall learning, with no diagnostic that distinguishes this from ordinary noise. Because the joint ratio offers zero upside over a per-agent ratio (Cor.~\ref{cor:canon}) and an unbounded, hard-to-anticipate downside that grows with $n$, the asymmetry of the bet is decisive: keep the ratio per-agent. This is especially pertinent at the scales where cooperative MARL is deployed—traffic grids, sensor and robot fleets, power networks—where $n$ is in the hundreds or thousands and the $(1+\chi^2)^n$ factor is one adverse batch away from biting. Our $196$-intersection traffic network bears this out directly: with $n=196$, the joint ratio collapses even at standard settings, while per-agent and neighborhood ratios learn normally (Appendix~\ref{app:traffic:ratio}).
\end{remark}

\section{Conclusion}

In this paper, we decomposed the credit-assignment structure of multi-agent policy optimization into two independent design axes: an advantage support and a ratio support. Our analysis reveals that the expected policy gradient depends on these supports only through their matrix product $\tS=\SR\SA$, establishing a gauge freedom that renders the two supports mutually redundant in terms of the gradient signal. This redundancy is broken by variance: aggregating through the advantage corresponds to a sum of rewards, yielding additive variance with an interior optimum at the coupling neighborhood, whereas aggregating through the ratio corresponds to a product of importance weights, incurring multiplicative, exponentially growing variance in the support size. These results lead to a clear and actionable design principle: aggregate neighbors in the advantage, sized according to the true reward coupling, and keep the ratio strictly per-agent. This rule applies universally to any factorized-policy PPO method and is supported by rigorous theoretical analysis. We validate our conclusions across four carefully designed cooperative MAMDP games and a real-world large-scale traffic signal control task, demonstrating both the correctness of the canonical form and the practical relevance of the variance asymmetry. More discussion is provided in Appendix~\ref{app:discussion}.


\bibliography{refs}
\bibliographystyle{iclr2026_conference}

\appendix
\section*{Appendix Contents}  
\startcontents  
\printcontents{}{1}{\setcounter{tocdepth}{2}}  
\newpage

\section{MAMDP Formulation}\label{app:mamdp}

We formalize the cooperative setting as a \emph{multi-agent Markov decision process} (MAMDP) \citep{littman1994markov}, equivalently a cooperative Markov game. This appendix provides the full formulation summarized in the main text.

\paragraph{Definition.} A MAMDP is a tuple
$\big(\mathcal N,\mathcal S,\{\mathcal A_i\}_{i\in\mathcal N},P,\{r_i\}_{i\in\mathcal N},\gamma,\mu_0\big)$,
where $\mathcal N=\{1,\dots,n\}$ is the set of agents; $\mathcal S$ is the state space; $\mathcal A_i$ is the finite action set of agent $i$, with joint action space $\mathcal A=\prod_{i}\mathcal A_i$ and joint action $\bm a=(a^1,\dots,a^n)$; $P(s'\mid s,\bm a)$ is the transition kernel; $r_i:\mathcal S\times\mathcal A\to\mathbb R$ is agent $i$'s reward, with team reward $\R(s,\bm a)=\sum_{i}r_i(s,\bm a)$; $\gamma\in[0,1)$ is the discount factor; and $\mu_0$ is the initial-state distribution. The setting is fully cooperative: all agents share the single team objective defined below.

\paragraph{Factorized policy.} Each agent acts through its own policy $\pi^i_{\bm\theta^i}(a^i\mid s)\in\Delta(\mathcal A_i)$ with a disjoint parameter block $\bm\theta^i$, so $\bm\theta=(\bm\theta^1,\dots,\bm\theta^n)$. Actions are conditionally independent across agents given the state, so the joint policy factorizes as
\begin{equation}\label{eq:factorize}
  \pi_{\bm\theta}(\bm a\mid s)=\prod_{i=1}^{n}\pi^i_{\bm\theta^i}(a^i\mid s).
\end{equation}
This factorization—equivalently, independent action heads sharing a state encoder—is the single structural assumption used in our analysis; it also covers the centralized multi-action policy that treats the whole team as one agent whose action is $\bm a$.

\paragraph{Objective.} The agents jointly maximize the expected discounted team return
\begin{equation}\label{eq:objective}
  J(\bm\theta)=\mathbb E_{s_0\sim\mu_0,\;\bm a_t\sim\pi_{\bm\theta}(\cdot\mid s_t),\;
  s_{t+1}\sim P(\cdot\mid s_t,\bm a_t)}\!\left[\sum_{t\ge0}\gamma^{t}\,\R(s_t,\bm a_t)\right].
\end{equation}

\paragraph{Values and advantages.} The state value is $V^{\pi}(s)=\mathbb E_{\pi}\!\big[\sum_{t\ge0}\gamma^{t}\R(s_t,\bm a_t)\mid s_0=s\big]$, and $Q^{\pi}$, the generalized advantage estimate (GAE), and the per-agent advantage are defined as usual from the per-agent rewards. The \emph{advantage support} $\SA\in\{0,1\}^{n\times n}$ selects which agents' rewards form agent $i$'s advantage: with $\hat A_j$ the single-agent advantage built from $r_j$, agent $i$'s aggregated advantage is $A_i=\sum_{j}\SA_{ij}\hat A_j$.

\paragraph{Coupling graph.} The \emph{coupling graph} $C\in\{0,1\}^{n\times n}$ has $C_{ij}=1$ iff $r_i$ depends on $a^j$ (with $C_{ii}=1$). Its smallest neighborhood radius $\rho^\star$ is the \emph{coupling radius}; $\partial i=\{j:C_{ij}=1\}$ is the coupling neighborhood of agent $i$.

\paragraph{Clipped surrogate with two supports.} Fixing a behavior policy $\bm\theta_{\mathrm{old}}$, the per-agent likelihood ratio is $\varrho^{j}=\pi^{j}_{\bm\theta}(a^j\mid s)/\pi^{j}_{\bm\theta_{\mathrm{old}}}(a^j\mid s)$. The \emph{ratio support} $\SR\in\{0,1\}^{n\times n}$ selects which agents' ratios enter agent $i$'s importance weight, $w_i=\prod_{j:\SR_{ij}=1}\varrho^{j}$, and the (clipped) multi-agent PPO surrogate is $\sum_i\mathbb E\big[\min(w_iA_i,\ \mathrm{clip}(w_i,1\pm\epsilon)A_i)\big]$. The pair $(\SA,\SR)$ is the object of study: $\SA=\SR=I$ corresponds to IPPO, $\SA=\one\one^\top,\SR=I$ to MAPPO, and $\SA=\SR=\one\one^\top$ to the fully centralized multi-action reduction.

\section{Proofs}\label{app:proofs}

Throughout, all expectations are over $\bm a\sim\pi_{\bm\theta_{\mathrm{old}}}$ unless noted, and we write $\varrho_j=\pi^j_{\bm\theta}(a^j)/\pi^j_{\bm\theta_{\mathrm{old}}}(a^j)$ for the per-agent likelihood ratio and $s^m:=\nabla_{\bm\theta^m}\log\pi^m(a^m)$ for the per-agent score.

\subsection{Proof of Lemma~\ref{lem:mean} (unbiased weight)}\label{app:proof:mean}
\begin{proof}
Fix $S\subseteq\mathcal N$. By Assumption~\ref{as:indep}, the actions $\{a^j\}_{j\in S}$ are independent under $\pi_{\bm\theta_{\mathrm{old}}}$, so the expectation of the product factorizes:
\begin{align}
\Eold\Big[\prod_{j\in S}\varrho_j\Big]
&=\prod_{j\in S}\Eold\big[\varrho_j\big]
 =\prod_{j\in S}\ \sum_{a\in\mathcal A}
   \pi^j_{\bm\theta_{\mathrm{old}}}(a)\,
   \frac{\pi^j_{\bm\theta}(a)}{\pi^j_{\bm\theta_{\mathrm{old}}}(a)}
\tag{10a}\label{eq:pf-mean-1}\\
&=\prod_{j\in S}\ \sum_{a\in\mathcal A}\pi^j_{\bm\theta}(a)
 =\prod_{j\in S} 1 = 1,
\tag{10b}\label{eq:pf-mean-2}
\end{align}
where \eqref{eq:pf-mean-1} uses independence (Assumption~\ref{as:indep}) and the definition of $\varrho_j$, and the cancellation in \eqref{eq:pf-mean-2} is valid because $\pi^j_{\bm\theta_{\mathrm{old}}}(a)>0$ (Assumption~\ref{as:supp}); the last equality follows from the normalization of $\pi^j_{\bm\theta}$. Taking $S=\{j:\SR_{ij}=1\}$ gives $\mathbb E[w_i]=1$ for every ratio support.
\end{proof}

\subsection{Proof of Lemma~\ref{lem:var} (multiplicative variance)}\label{app:proof:var}
\begin{proof}
By Lemma~\ref{lem:mean}, $\Eold[\prod_{j\in S}\varrho_j]=1$, so $\Var(\prod_{j\in S}\varrho_j)=\Eold[(\prod_{j\in S}\varrho_j)^2]-1$. The squared product factorizes over agents by independence (Assumption~\ref{as:indep}):
\begin{align}
&\Eold\Big[\Big(\prod_{j\in S}\varrho_j\Big)^2\Big]
=\prod_{j\in S}\Eold\big[\varrho_j^2\big], \\
\qquad\text{and}\qquad
&\Eold\big[\varrho_j^2\big]
=\sum_{a\in\mathcal A}\pi^j_{\bm\theta_{\mathrm{old}}}(a)
  \frac{\pi^j_{\bm\theta}(a)^2}{\pi^j_{\bm\theta_{\mathrm{old}}}(a)^2}
=\sum_{a}\frac{\pi^j_{\bm\theta}(a)^2}{\pi^j_{\bm\theta_{\mathrm{old}}}(a)} .
\label{eq:pf-var-1}
\end{align}
The rightmost sum is $1+\chi^2\!\big(\pi^j_{\bm\theta}\,\|\,\pi^j_{\bm\theta_{\mathrm{old}}}\big)=1+\chi^2_j$ by the definition of the $\chi^2$-divergence. Substituting this into \eqref{eq:pf-var-1} and subtracting $1$ gives \eqref{eq:varprod}: $\Var(\prod_{j\in S}\varrho_j)=\prod_{j\in S}(1+\chi^2_j)-1$. Since each factor $1+\chi^2_j\ge1$, the product is nondecreasing in $S$; if $\chi^2_j\ge c>0$ for all $j$, then $\prod_{j\in S}(1+\chi^2_j)-1\ge(1+c)^{|S|}-1$, which is exponential in $|S|$.
\end{proof}

\subsection{Proof of Theorem~\ref{thm:canon} (support factorization)}\label{app:proof:canon}
\begin{proof}
Evaluate the gradient \eqref{eq:grad} at $\bm\theta=\bm\theta_{\mathrm{old}}$, where every $\varrho_j=1$ and hence $w_i=1$. Writing $A_i=\sum_j\SA_{ij}r_j-b_i$,
\begin{align}
\bm g_m
&=\Eold\Big[\Big(\sum_{i:\SR_{im}=1}A_i\Big)\,s^m\Big]
 =\sum_{i:\SR_{im}=1}\ \sum_{j}\SA_{ij}\,\Eold\big[r_j\,s^m\big]
 \;-\;\sum_{i:\SR_{im}=1} b_i\,\Eold\big[s^m\big].
\label{eq:pf-thm-1}
\end{align}
The baseline term vanishes: $b_i$ is constant in $\bm a$ and $\Eold[s^m]=\sum_{a^m}\pi^m_{\bm\theta_{\mathrm{old}}}(a^m)\nabla_{\bm\theta^m}\log\pi^m(a^m)=\nabla_{\bm\theta^m}\sum_{a^m}\pi^m(a^m)=\nabla_{\bm\theta^m}1=0$. Next, by Assumption~\ref{as:sym} (symmetric coupling) and conditional independence, $\Eold[r_j\,s^m]=0$ whenever $C_{jm}=0$: if $r_j$ does not depend on $a^m$, conditioning on $a^{-m}$ gives
\begin{align}
\Eold\big[r_j\,s^m\big]
=\Eold\Big[r_j\,\Eold\big[s^m\mid a^{-m}\big]\Big]
=\Eold\big[r_j\big]\cdot\Eold\big[s^m\big]=0,
\label{eq:pf-thm-2}
\end{align}
using $\Eold[s^m\mid a^{-m}]=\Eold[s^m]=0$ (action independence). Substituting \eqref{eq:pf-thm-2} into \eqref{eq:pf-thm-1} and exchanging the order of the $i$ and $j$ sums,
\begin{align}
\bm g_m
=\sum_{j:\,C_{jm}=1}\Big(\sum_{i}\SR_{im}\SA_{ij}\Big)\Eold\big[r_j\,s^m\big]
=\sum_{j:\,C_{jm}=1}\big(\SR\SA\big)_{mj}\,\Eold\big[r_j\,s^m\big],
\label{eq:pf-thm-3}
\end{align}
where the last step uses $\SR_{im}=\SR_{mi}$ (symmetry), so $\sum_i\SR_{mi}\SA_{ij}=(\SR\SA)_{mj}=\tS_{mj}$. This is \eqref{eq:tildeS}. Since $\tS=\SR\SA$ is the only way $(\SA,\SR)$ enter, any two pairs with the same product induce the same $\bm g_m$.
\end{proof}

\subsection{Proof of Corollary~\ref{cor:canon} (per-agent-ratio canonical form)}\label{app:proof:corcanon}
\begin{proof}
The estimator $(\SR=I,\SA=\tS)$ has product $\tS'=I\cdot\tS=\tS$, identical to that of $(\SA,\SR)$. By Theorem~\ref{thm:canon}, the expected gradient depends on the supports only through this product, so the two estimators share the same $\bm g_m$; concretely, reading from \eqref{eq:tildeS}, both equal $\sum_{j:C_{jm}=1}\tS_{mj}\Eold[r_j s^m]$, which is exactly the gradient of a per-agent-ratio estimator with advantage $A'_m=\sum_j\tS_{mj}r_j$ (a nonnegative integer-weighted reward sum, hence admissible). For the joint ratio $\SR=\one\one^\top$, the reweighting is $\tS=\one\one^\top\SA$. Thus, cross-agent ratios contribute nothing to $\bm g_m$ beyond the linear advantage reweighting $r\mapsto\tS r$.
\end{proof}

\subsection{Proof of Proposition~\ref{prop:dom} (off-policy variance domination)}\label{app:proof:dom}
\begin{proof}
Both paths in \eqref{eq:pq} realize the same product $\tS=C$, so by Theorem~\ref{thm:canon} they share the mean $\bm g_m=\sum_{j\in\partial m}\Eold[r_j s^m]$. Write the per-step weights whose product with $s^m$ is averaged: for the advantage path (P), $W^{\mathrm P}_m=\sum_{j\in\partial m} r_j$; for the ratio path (Q), $W^{\mathrm Q}_m=\sum_{i\in\partial m}\big(\prod_{j\in\partial i}\varrho_j\big) r_i$.

\emph{(i) On-policy.} At $\bm\theta=\bm\theta_{\mathrm{old}}$, every $\varrho_j=1$, so $W^{\mathrm Q}_m=\sum_{i\in\partial m} r_i=W^{\mathrm P}_m$ pathwise; the two estimators are identical random variables and share mean and variance.

\emph{(ii) Off-policy.} $W^{\mathrm P}_m$ has no $\bm\theta$ dependence and is bounded: $|W^{\mathrm P}_m|\le|\partial m|\,r_{\max}$ (Assumption~\ref{as:bdd}), so $\Var(W^{\mathrm P}_m)$ is a fixed constant. For Q, take a single summand $i\in\partial m$ and condition on $s^m$; by Assumption~\ref{as:bdd}, there is $\delta>0$ with $\Eold[r_i^2\mid a^m]\ge\delta$ on a set of positive probability, and by independence of the ratios from $r_i$ and Lemma~\ref{lem:var},
\begin{align}
\Eold\Big[\Big(\prod_{j\in\partial i}\varrho_j\Big)^2\,r_i^2\Big]
\;\ge\;\delta\,\Eold\Big[\Big(\prod_{j\in\partial i}\varrho_j\Big)^2\Big]
\;=\;\delta\prod_{j\in\partial i}(1+\chi^2_j).
\label{eq:pf-dom-1}
\end{align}
Hence the second moment of the Q-weight, and therefore $\Var(\bm g^{\mathrm Q}_m)$, is bounded below by a positive multiple of $\prod_{j\in\partial i}(1+\chi^2_j)$, which grows without bound as the policy step increases (each $\chi^2_j\to\infty$) or as $|\partial m|$ grows (more factors). Since $\Var(\bm g^{\mathrm P}_m)$ stays bounded, $\Var(\bm g^{\mathrm Q}_m)/\Var(\bm g^{\mathrm P}_m)\to\infty$.
\end{proof}

\subsection{Proof of Proposition~\ref{prop:adv} (advantage bias--variance tradeoff)}\label{app:proof:adv}
\begin{proof}
Take $\SR=I$, so $\tS=\SA$ and by Theorem~\ref{thm:canon} $\bm g_m=\sum_{j:C_{jm}=1}\SA_{mj}\,\Eold[r_j s^m]$. The unbiased (full-coupling) gradient is $\bm g^\star_m=\sum_{j\in\partial m}\Eold[r_j s^m]$, i.e., the case $\SA_{mj}=1$ for all $j\in\partial m$.

\emph{Bias.} Subtracting,
\begin{align}
\bm g_m-\bm g^\star_m
=\sum_{j\in\partial m}(\SA_{mj}-1)\,\Eold[r_j s^m]
=-\!\!\sum_{j\in\partial m:\ \SA_{mj}=0}\!\!\Eold[r_j s^m],
\label{eq:pf-adv-1}
\end{align}
so the bias is exactly the sum of the score--reward couplings of the coupled agents omitted by $\SA$. It is zero iff $\SA\supseteq\partial m$ and strictly nonzero when a genuinely coupled agent with $\Eold[r_j s^m]\ne0$ is dropped. Adding agents outside $\partial m$ changes nothing, since $\Eold[r_j s^m]=0$ there by \eqref{eq:pf-thm-2}.

\emph{Variance.} The aggregated advantage $A_m=\sum_j\SA_{mj}r_j$ is additive. Including an extra agent $j$ adds the term $r_j s^m$ to the estimator. This term has mean $\Eold[r_j s^m]$ and contributes
\begin{align}
\Delta\Var
=\Var\big(r_j s^m\big)+2\,\mathrm{Cov}\Big(r_j s^m,\ \textstyle\sum_{k\in\SA_m\setminus j} r_k s^m\Big),
\label{eq:pf-adv-2}
\end{align}
whose dominant, always-nonnegative part is $\Eold[r_j^2 (s^m)^2]\ge0$ (strictly positive when $r_j$ is conditionally nondegenerate). Thus each added agent raises the variance, while only agents in $\partial m$ reduce the bias \eqref{eq:pf-adv-1}. The MSE $\|\bm g_m-\bm g^\star_m\|^2+\Var(\bm g_m)$ is therefore minimized by including exactly the coupled agents, $\SA=\{\,j:C_{jm}=1\,\}=\partial m$: smaller supports are biased, larger supports inflate variance.
\end{proof}

\section{Experimental Details and Additional Results}\label{app:expmore}

\subsection{Experimental Setup}\label{app:setup}

\subsubsection{Choice of Environments}

Our study requires environments with two properties that standard cooperative benchmarks lack: a per-agent reward (so the advantage support is meaningful) and a known coupling graph (so the coupling neighborhood is defined). Popular suites such as StarCraft \citep{whiteson2019starcraft} and Google Research Football \citep{kurach2020google} expose a single shared team reward and no explicit coupling structure; verifying the advantage-support axis there would require synthesizing per-agent signals via difference rewards, introducing approximation and forfeiting the ground-truth coupling graph. We therefore study synthetic cooperative games in which the coupling graph is prescribed exactly, allowing us to vary the coupling family, its range, and the agent count while keeping the ground truth fixed. Applying the same analysis to real systems that natively provide per-agent rewards and a physical coupling graph—traffic-signal control being a canonical example—is a natural direction we leave to future work. Our goal here is a clean verification of the mechanism, not benchmark leaderboard performance, so the environments are simple and fully specified.

The families, shown as Fig. \ref{fig:families}, are chosen to separate the two axes the theory distinguishes, and to stress each in turn. They differ first in whether an agent internalizes its own externality: in the directed dilemma, an agent's action harms only its successors, not itself, so an independent advantage is badly biased and the advantage support is decisive for the learned policy; in dense pairwise and local congestion, the coupling term enters the actor's own reward (fully or partly), so an independent advantage is already near-unbiased and the advantage support mainly trades variance. This spread lets us show that the advantage radius matters for the outcome exactly when the externality is not self-internalized (Fig.~\ref{fig:advsweep}). The families differ second in graph structure: the ring families (dilemma, pairwise, congestion) provide a geometric coupling with a well-defined radius, while block community gives a non-geometric, community-structured graph—verifying that matching the advantage support to the coupling neighborhood is about the coupling graph, not about spatial distance. The ratio-support conclusions (redundancy and variance ordering), by contrast, depend only on the number of factors and hold across all four, which is why we use every family for the ratio studies but highlight the dilemma for the advantage-outcome effect.

\subsubsection{Coupling Families: Reward Definitions}

All environments used in this paper are cooperative games with a prescribed coupling graph $C$, so the ground-truth coupling neighborhood is known exactly. For completeness, we briefly restate the notation. There are $n$ agents; each agent $i$ observes a phase or context and selects a discrete action $a^i \in \{1,\dots,K\}$, where $K$ is the number of actions per agent. We write the joint action as $\bm a = (a^1,\dots,a^n)$. Each agent receives a per-agent reward $r_i(\bm a)$, defined per family below, and the team return is their sum $\R = \sum_i r_i$. The coupling graph $C \in \{0,1\}^{n \times n}$ records reward dependencies, with $C_{ij}=1$ iff $r_i$ depends on $a^j$; its neighborhood $\partial i = \{j : C_{ij}=1\}$ is the set of agents whose actions enter $r_i$. For the ring families, $\partial i$ consists of the agents within graph distance $\rho^\star$ on a ring, and we refer to $\rho^\star$ as the coupling radius. Zero-mean Gaussian observation noise $\varepsilon_i \sim \mathcal N(0,\sigma^2)$ with standard deviation $\sigma$ is added where noted.

\begin{figure}[t]
\centering
\includegraphics[width=\textwidth]{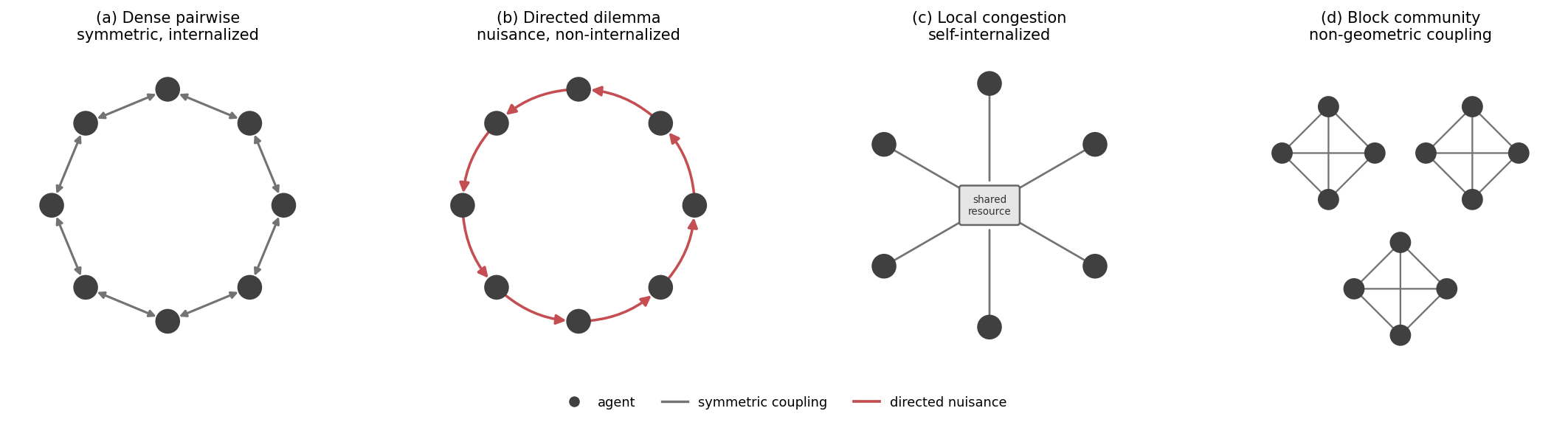}
\caption{Schematic of the four coupling families (illustrated at a reduced scale; nodes represent agents, edges represent reward dependencies). \textbf{(a)} Dense pairwise: undirected edges couple ring neighbors symmetrically, so that each agent's action affects its neighbors as much as itself. \textbf{(b)} Directed dilemma: each agent transmits a one-way nuisance (red) to its downstream successors, which the actor does not bear. \textbf{(c)} Local congestion: neighbors selecting the same resource share the reward. \textbf{(d)} Block community: all-to-all coupling within disjoint blocks, with no coupling across blocks.}
\label{fig:families}
\end{figure}

\paragraph{Dense pairwise (symmetric, internalized).} This family uses a symmetric pairwise payoff over the ring neighborhood,
\begin{equation}\label{eq:densepair}
r_i(\bm a)=\sum_{j\in\partial i,\,j\ne i} M[a^i,a^j]+\varepsilon_i,\qquad
M=M^\top\in\mathbb R^{K\times K},\ \varepsilon_i\sim\mathcal N(0,\sigma^2),
\end{equation}
where $M$ is a fixed symmetric $K\times K$ payoff matrix whose entry $M[a^i,a^j]$ is the payoff agent $i$ obtains from the action pair $(a^i,a^j)$ with neighbor $j$, and $\varepsilon_i$ is the observation noise defined above. Because $M$ is symmetric, the same edge $M[a^i,a^j]$ contributes to both $r_i$ and $r_j$; hence an agent's action affects its neighbors' reward as strongly as its own—the externality is internalized. This game supplies the single-step results in Fig.~\ref{fig:motiv} ($n=20$, $K=4$, $\rho^\star=4$, $\sigma=0.6$); the training study uses $n=24$.

\paragraph{Directed dilemma (non-internalized externality).} On a directed ring, choosing a high-value action gives the actor a private benefit but imposes a nuisance on its $\rho^\star$ successors:
\begin{equation}\label{eq:dilemma}
r_i(\bm a)=v(a^i)-c\!\!\sum_{j:\,i\in\partial^{+}j}\! g(a^j),
\end{equation}
where $v(a^i)$ is the private value agent $i$ obtains from its own action, $g(a^j)$ is the nuisance that agent $j$'s action imposes on its downstream neighbors, $c>0$ is the weight of that nuisance, and $\partial^{+}j$ denotes the $\rho^\star$ successors of $j$ on the directed ring; the sum thus runs over those agents $j$ whose successor set contains $i$, i.e., the predecessors that emit onto $i$. A rotating phase makes the tempting action time-varying. Since the actor does not bear the harm it causes, an independent advantage converges to the tragedy-of-the-commons equilibrium, whereas a neighborhood advantage recovers the social optimum. We use $n=20$, $\rho^\star=2$, $c=3$.

\paragraph{Local congestion (self-internalized).} Agents within $\partial i$ that choose the same resource split its value:
\begin{equation}\label{eq:congestion}
r_i(\bm a)=\frac{v(a^i)}{\big|\{j\in\partial i: a^j=a^i\}\big|}+\varepsilon_i,
\end{equation}
where $v(a^i)$ is the base value of the resource selected by action $a^i$, and the denominator counts the number of agents in the neighborhood $\partial i$ (including $i$ itself) that chose the same resource, so the resource value is shared equally among them. An agent that over-subscribes a popular resource thus depresses its own reward as well as its neighbors'—the crowding cost is partly self-borne. We use $n=20$, $\rho^\star=2$.

\paragraph{Block community (non-geometric coupling).} Agents are partitioned into disjoint blocks; coupling is all-to-all within a block and absent across blocks, so $C$ is block-diagonal, with an intra-block symmetric payoff matrix $M$ as in the dense-pairwise case \eqref{eq:densepair} (here $\partial i$ is $i$'s entire block rather than a ring neighborhood). This setup tests a coupling graph that is not a geometric ring. We use $n=20$ and block size $5$.

In every case, the reward parameters ($M$, $v$, $g$, block assignment) are fixed per environment; we report mean $\pm$ standard deviation over random seeds. Full hyperparameters are deferred to the appendix.

\subsection{Extended Training-Curve Analysis}\label{app:figs}

This subsection reports the full per-family training dynamics summarized in the main text: the advantage-support curves (complementing the sweep in Fig.~\ref{fig:advsweep}) and the ratio-support redundancy curves. Together, they demonstrate, family by family, that enlarging the advantage support can change the learned policy, whereas enlarging the ratio support never does.

Figure~\ref{fig:traincurves} plots the team return over PPO iterations for each family under three advantage supports—independent ($\rho_A=0$), coupling-radius, and joint—with the ratio held per-agent. The directed dilemma shows a dramatic separation: independent advantage descends to the tragedy equilibrium, while coupling-radius and joint advantage climb to the social optimum. In the other three families, where the externality is self-internalized, independent advantage already learns effectively, and larger supports add variance rather than altering the outcome. This illustrates the advantage side of the asymmetry: the advantage support can shift the learned policy.

To demonstrate that the ratio radius is redundant, we fix the team advantage and sweep the ratio support $\rho_R \in \{\text{per-agent}, \text{neighborhood}, \text{joint}\}$ (Fig.~\ref{fig:redundancy}). Within every family, the three learning curves are indistinguishable—overlapping to within seed noise, with clip fraction $0$ throughout—so enlarging the ratio support changes neither the trajectory nor the final outcome, exactly as predicted by Corollary~\ref{cor:canon}. This is the ratio-side counterpart to Fig.~\ref{fig:traincurves}: enlarging the advantage support can change the outcome, whereas enlarging the ratio support never does. 

\begin{figure}[t]
\centering
\includegraphics[width=\textwidth]{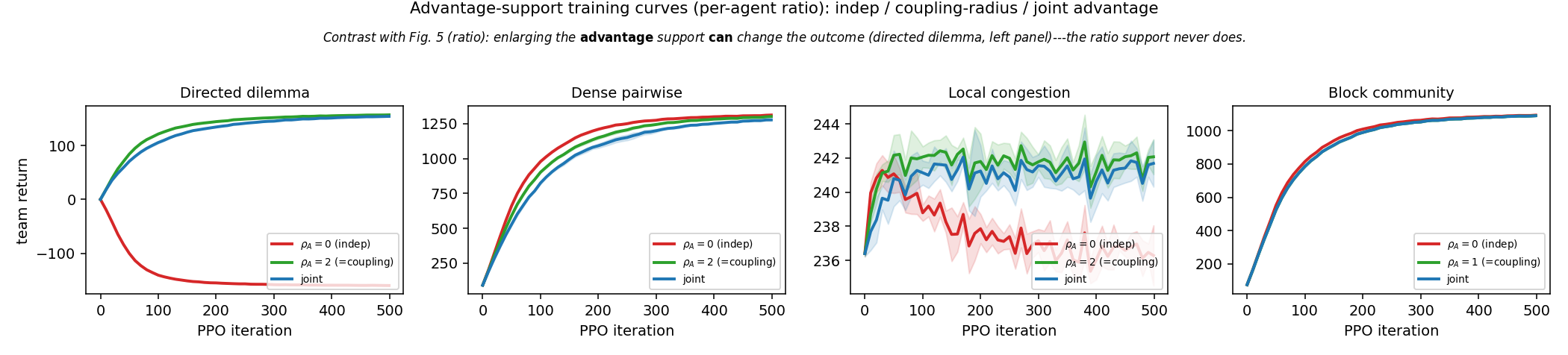}
\caption{Advantage-support training curves across the four coupling families. Team return versus PPO iteration ($5$ seeds, $\pm1$ std bands) under three advantage supports: independent ($\rho_A=0$), coupling-radius, and joint, with the ratio held per-agent. One panel per family. In the directed-dilemma and dense-pairwise panels, the seed variance is below $1.5\%$ of the return scale, so the bands are narrower than the line width.}
\label{fig:traincurves}
\end{figure}

\begin{figure}[t]
\centering
\includegraphics[width=\textwidth]{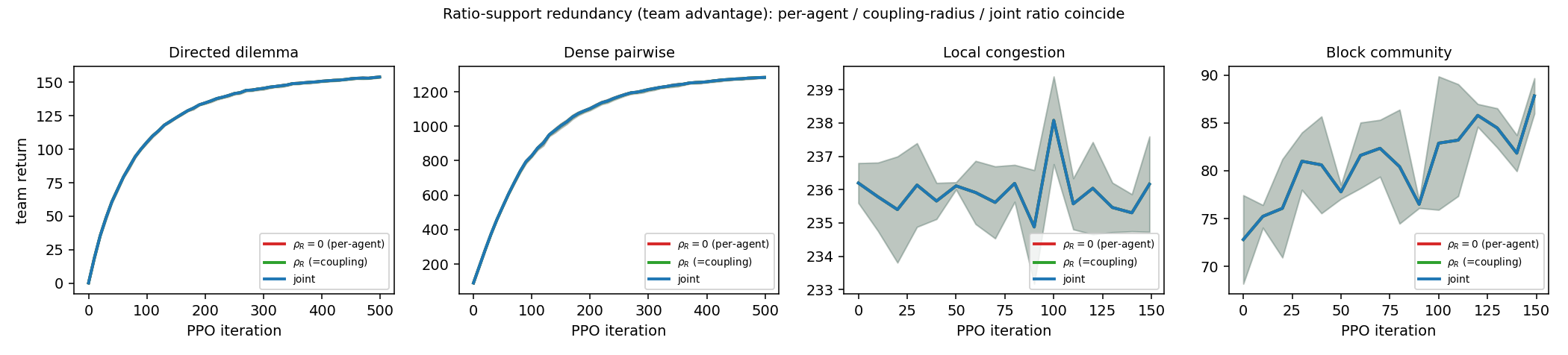}
\caption{Ratio-support redundancy under standard PPO across the four coupling families. Team return versus PPO iteration ($\pm1$ std bands over seeds) with team advantage fixed. Each panel shows three ratio supports: per-agent ($\rho_R=0$), coupling-radius, and joint.}
\label{fig:redundancy}
\end{figure}

\subsection{Real-World Traffic-Signal Control}\label{sec:exp:traffic}

Finally, we test the advantage-support prediction (Proposition~\ref{prop:adv}) on a real cooperative traffic-signal control problem that natively provides both ingredients the theory requires: a per-agent reward (each intersection's local queue, delay, and throughput) and a physical coupling graph (the road network). We use the SUMO \citep{behrisch2011sumo} Manhattan $28\times7$ grid ($n=196$ traffic signals) via \texttt{sumo-rl} \citep{sumorl}, with the standard MLP-actor MAPPO from the toy suite; the only tunable parameter is the advantage support, built by $k$-hop expansion on the road-adjacency graph, with a per-agent ratio throughout. 

Each intersection $i$ is an agent with a local observation $o_i$ (the standard \texttt{sumo-rl} encoding): a one-hot of the active green phase, a binary flag indicating whether the minimum green time has elapsed, and, for every incoming lane, the normalized vehicle density and queue length. The action $a_i$ is discrete—the choice of the next green phase from that intersection's signal program—applied on a fixed $10$\,s control cycle with a $2$\,s yellow transition and a $5$\,s/$50$\,s minimum/maximum green. The shared MLP actor maps $o_i$ to a categorical distribution over $a_i$; agents share parameters but act independently, so the joint policy factorizes as $\prod_i\pi^i(a_i)$, exactly as our analysis assumes. Each agent's reward is a purely local combination of its own traffic state:
\begin{equation}\label{eq:trafficreward}
  r_i \;=\; w_v\,\bar v_i \;-\; w_q\,q_i \;-\; w_\omega\,\omega_i \;-\; w_p\,p_i,
  \qquad (w_v,\,w_q,\,w_\omega,\,w_p)=(3.0,\,1.0,\,0.3,\,1.5),
\end{equation}
where $q_i$ is the total queue length, $\omega_i$ the total waiting time, $p_i$ the inflow--outflow pressure $|\mathrm{in}-\mathrm{out}|$, and $\bar v_i$ the mean speed, each normalized per lane (by $50$, $1000$, $20$ vehicles, and $15$\,m/s, respectively). Here $r_i$ depends on other intersections only through physical spillover onto adjacent roads—so the reward coupling graph is exactly the road-network adjacency (each signal coupled to its up-to-four grid neighbors), and the range $\rho^\star$ is what our advantage support aims to match.

\subsubsection{Advantage-Support Study: Global Advantage Collapses; Local Advantages Work (Fig.~\ref{fig:traffic})}

Figure~\ref{fig:traffic} plots the three traffic metrics—team reward (the trained objective), average queue length, and average waiting time—over PPO iterations, comparing the global (team) advantage against three local advantage supports; Table~\ref{tab:traffic} reports the converged metrics. A team advantage—the radius-$\infty$ support that aggregates all $196$ intersections' rewards, as in vanilla MAPPO—fails to learn: its team return remains flat near its initial value while local advantages climb steadily, and at convergence it is $36\%$ worse in return, with $2.4\times$ the average queue and $3.7\times$ the average waiting time (Table~\ref{tab:traffic}). This is Proposition~\ref{prop:adv}(ii) at scale: summing $196$ per-intersection rewards drowns each signal's own learning contribution in the noise of $195$ others. In contrast, the three local supports—independent ($\rho_A=0$) and one/two-hop neighborhoods—all learn well and closely track one another, with a mild monotone gain from including immediate neighbors; the insets confirm that these three curves differ only marginally and within the seed bands. The small spread among local radii indicates that traffic coupling is largely self-internalized (an intersection's own queue strongly reflects its own action), placing this system in the regime where the advantage-support bias is small and the dominant effect is the variance cost of an over-large support—exactly the failure mode exhibited by the team advantage. The team--local gap far exceeds the seed spread.

\begin{figure}[t]
\centering
\includegraphics[width=\textwidth]{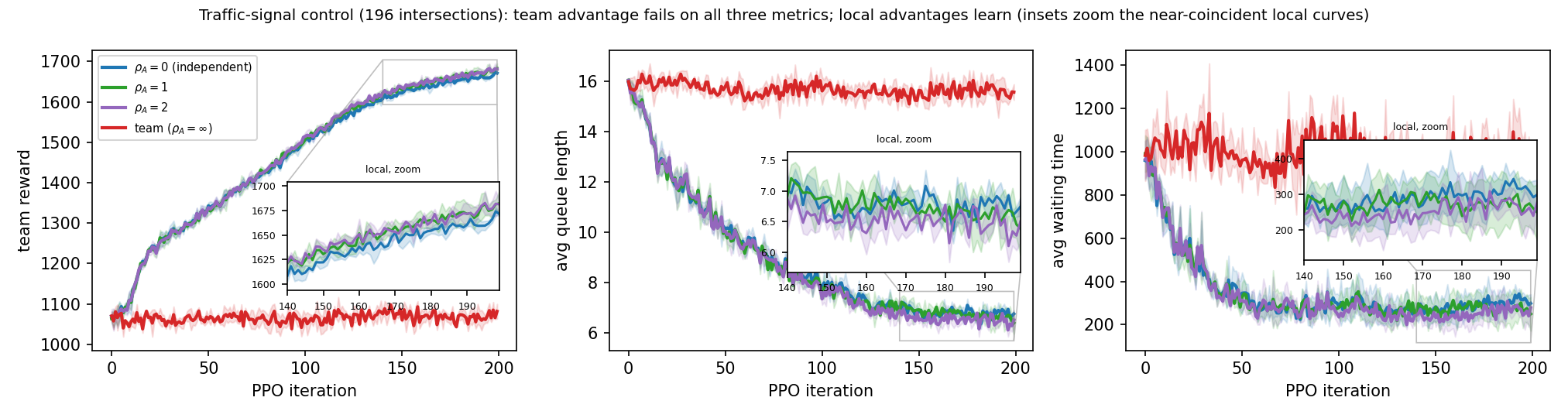}
\caption{Traffic-signal control on a $196$-intersection network. Advantage-support training curves (per-agent ratio; mean over $3$ seeds, $\pm1$ std shaded bands) on all three metrics—team reward, average queue, and average waiting time—versus PPO iteration. Each panel compares three local advantage supports (independent ($\rho_A=0$) and one/two-hop neighborhoods) against the global/team advantage ($\rho_A=\infty$). Insets zoom the last $60$ iterations for the three local supports alone ($\rho_A=0,1,2$).}
\label{fig:traffic}
\end{figure}

\begin{table}[t]
\centering
\caption{Traffic-signal control, $196$ intersections, final metrics (mean $\pm$ std over $3$ seeds, last $20\%$ of training). Local advantages show strong performance; the global (team) advantage collapses.}
\label{tab:traffic}
\begin{tabular}{lcccc}
\toprule
advantage support & team return $\uparrow$ & avg queue $\downarrow$ &
avg wait (s) $\downarrow$ & clip frac \\
\midrule
independent ($\rho_A=0$)   & $1652.8{\pm}4.5$ & $6.77{\pm}0.04$  & $298.0{\pm}47.1$ & $0.014$ \\
$1$-hop neighborhood         & $1663.7{\pm}7.5$ & $6.68{\pm}0.21$  & $275.9{\pm}57.7$ & $0.016$ \\
$2$-hop neighborhood         & $1663.8{\pm}6.0$ & $6.50{\pm}0.06$  & $258.6{\pm}45.4$ & $0.016$ \\
global / team ($\rho_A=\infty$) & $1068.1{\pm}7.5$ & $15.62{\pm}0.23$ & $968.9{\pm}48.7$ & $0.000$ \\
\bottomrule
\end{tabular}
\end{table}

\subsubsection{Ratio-Support Study: Redundancy and Variance Ordering}
\label{app:traffic:ratio}

The advantage study above fixes a per-agent ratio and varies the advantage support. We now perform the converse: vary the ratio support—per-agent ($\SR=I$), $1$-hop and $2$-hop neighborhoods, and joint ($\SR=\one\one^\top$, the $196$-way product)—and, to verify that the ratio-side conclusion does not depend on how the advantage is chosen, we repeat the entire sweep under two fixed advantage baselines: the $1$-hop neighborhood advantage and the independent ($\SA=I$) advantage. All runs use a standard PPO configuration. This isolates the ratio support on a real system and tests the two ratio-side predictions—redundancy (Corollary~\ref{cor:canon}) and variance ordering that worsens with $n$ (Proposition~\ref{prop:dom}, Fig.~\ref{fig:scaling}).

Figure~\ref{fig:traffic_ratio} plots, for each advantage baseline (rows) and each metric (columns), the training curve of every ratio support. The three small ratio supports—per-agent, $1$-hop, and $2$-hop—learn equally well and reach essentially the same reward, queue, and waiting time: enlarging the ratio support from $0$ to $2$ hops changes nothing about the learned policy, exactly as the redundancy corollary (Corollary~\ref{cor:canon}) predicts. The joint $196$-way ratio, in contrast, collapses outright and it does so whether the advantage is the $1$-hop neighborhood or fully independent. This is the sharpest real-system signature of the $n$-dependence in Fig.~\ref{fig:scaling}: with $n=196$ agents, the joint ratio is a $196$-way product whose variance factor $\prod_j(1+\chi^2_j)$ is enormous, drowning the surrogate gradient in variance well before any benefit could accrue. Crucially, no off-policy stress is needed to expose this—unlike the toy games, where $n\le24$ keeps the standard-configuration penalty dormant, the traffic network is large enough that the joint ratio fails under an ordinary, conservative PPO configuration. That the collapse and the redundancy are identical across the two advantage baselines confirms that the ratio-side conclusion is a property of the ratio support alone, independent of the advantage support with which it is paired.

\begin{figure}[t]
\centering
\includegraphics[width=\textwidth]{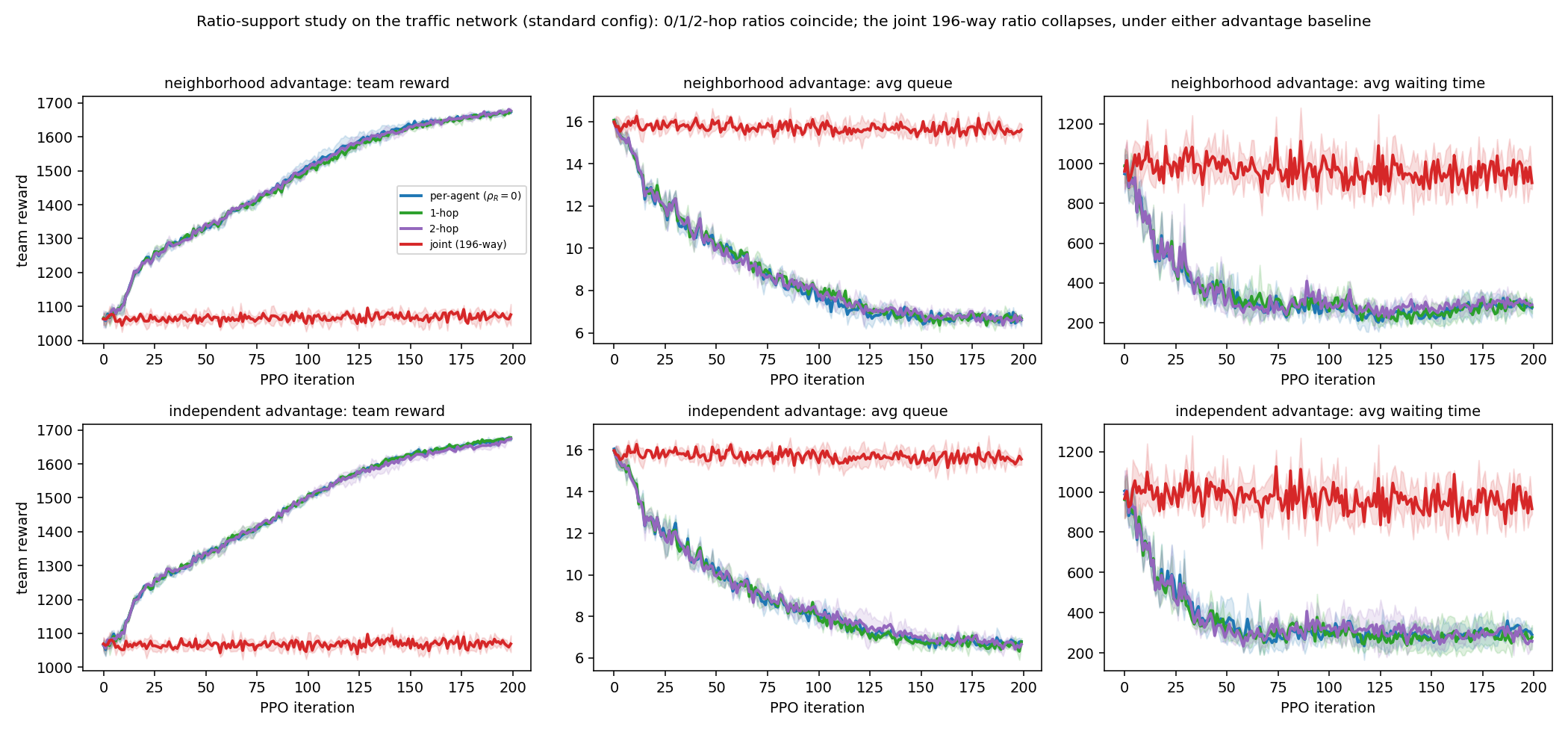}
\caption{Ratio-support study on the $196$-intersection network under a standard PPO configuration (mean over $3$ seeds, $\pm1$ std bands). Top row: the $1$-hop neighborhood advantage held fixed. Bottom row: the independent ($\SA=I$) advantage held fixed. Columns show team reward, average queue, and average waiting time versus PPO iteration. Within each panel, four curves vary the ratio support: per-agent ($\rho_R=0$), $1$-hop, $2$-hop, and joint ($196$-way).}
\label{fig:traffic_ratio}
\end{figure}

\begin{table}[t]
\centering
\caption{Ratio-support study on the traffic network, final metrics (mean $\pm$ std over $3$ seeds, last $20\%$), under two fixed advantage baselines. Ratio supports of $0$, $1$, and $2$ hops learn; the joint $196$-way ratio collapses under either baseline.}
\label{tab:traffic_ratio}
\begin{tabular}{llcccc}
\toprule
advantage & ratio support & team return $\uparrow$ & avg queue $\downarrow$ &
avg wait (s) $\downarrow$ & clip frac \\
\midrule
\multirow{4}{*}{$1$-hop nbr}
 & per-agent ($\rho_R=0$) & $1659.5{\pm}4.6$ & $6.69{\pm}0.04$ & $286.0{\pm}39.0$ & $0.014$ \\
 & $1$-hop                  & $1657.1{\pm}2.7$ & $6.71{\pm}0.06$ & $280.3{\pm}36.6$ & $0.016$ \\
 & $2$-hop                  & $1659.4{\pm}6.2$ & $6.72{\pm}0.15$ & $295.8{\pm}29.2$ & $0.016$ \\
 & joint ($196$-way)        & $1071.1{\pm}6.7$ & $15.63{\pm}0.18$ & $950.6{\pm}37.1$ & $0.041$ \\
\midrule
\multirow{4}{*}{independent}
 & per-agent ($\rho_R=0$) & $1656.5{\pm}6.5$ & $6.77{\pm}0.09$ & $301.4{\pm}42.9$ & $0.015$ \\
 & $1$-hop                  & $1658.6{\pm}1.5$ & $6.68{\pm}0.13$ & $286.4{\pm}60.3$ & $0.016$ \\
 & $2$-hop                  & $1649.7{\pm}6.1$ & $6.79{\pm}0.12$ & $291.0{\pm}8.1$ & $0.015$ \\
 & joint ($196$-way)        & $1070.6{\pm}7.0$ & $15.63{\pm}0.21$ & $945.8{\pm}46.1$ & $0.041$ \\
\bottomrule
\end{tabular}
\end{table}

\section{Related Work}\label{sec:related}

\paragraph{Trust-region and proximal MARL.}
Extending PPO and TRPO to cooperative MARL requires deciding how the importance ratio is formed across agents. MAPPO \citep{yu2022mappo} pairs a shared team advantage with a per-agent ratio and serves as a strong benchmark baseline; IPPO trains each agent independently, also with a per-agent ratio. At the other extreme, HATRPO and HAPPO \citep{kuba2022trust} update agents sequentially and clip a compound ratio—a product of the ratios of previously updated agents—justified by the multi-agent advantage-decomposition lemma. These methods span the ratio-support axis we study, from the per-agent extreme ($\SR=I$) to the fully joint extreme ($\SR=\one\one^\top$), yet the choice is typically made by architectural convention rather than analysis. Our results speak directly to this decision: the compound or joint ratio is gradient-redundant relative to a per-agent ratio with a suitably reweighted advantage (Theorem~\ref{thm:canon}), and incurs strictly higher variance (Proposition~\ref{prop:dom}).

\paragraph{Locality and scalability in networked MARL.}
A parallel line of work exploits the observation that, in networked systems, an agent's influence decays with graph distance. \cite{qu2020scalable,qu2022scalable} formalize this as an exponential decay property and use it to prove that truncating each agent's value estimate to a $\kappa$-hop neighborhood incurs only $O(\rho^\kappa)$ error, yielding provably scalable actor--critic algorithms. This provides the theoretical grounding for our advantage-support axis: if reward coupling has finite range $\rho^\star$, then aggregating rewards over the $\rho^\star$-neighborhood is nearly unbiased, while a smaller support omits real externalities and a larger one only adds variance. We complement this literature by (i) separating the advantage support from the ratio support and (ii) showing that the two are redundant on the expected gradient—questions orthogonal to value-function truncation.

\paragraph{Counterfactual and difference credit.}
COMA \citep{foerster2018coma} and difference-reward methods \citep{wolpert2001coin,li2022difference} reduce variance by subtracting a counterfactual baseline whose expected gradient contribution is zero; recent per-agent advantage estimators \citep{genadv2026} prove policy-gradient invariance of such baselines under the factorized policy. These works concern the advantage or baseline and its zero-offset property; none addresses the ratio support or the factorization of the expected gradient through $\tS=\SR\SA$, which is the object of our study.

\paragraph{Local advantages and local ratios.}
Several applied methods already restrict the advantage to a neighborhood without a general analysis. In traffic-signal control, MA2C \citep{chu2020ma2c} introduces a spatial discount factor that down-weights distant agents' rewards inside each local return—an advantage support tapered by graph distance. We observe an asymmetry in this literature: a rich set of methods aggregate the advantage locally, but we are not aware of any cooperative method that aggregates the ratio locally (e.g., a local importance-sampling weight $\prod_{j\in\partial i}\varrho^j$). Where a joint ratio does appear, it is typically not chosen for its own sake but inherited from casting the whole team as a single multi-action agent: running PPO on the factorized joint policy $\prod_i\pi^i$ treats $\bm a$ as one action and thus clips the product of all agents' ratios. This reduction is attractive because it turns MARL into single-agent RL—with its simpler objective and convergence story—and is used, for example, to make combinatorial cooperative tasks tractable as a single-controller problem. Our results explain the resulting high variance: the two supports are interchangeable on the expected gradient (Theorem~\ref{thm:canon}), so the joint ratio buys nothing over a per-agent ratio, yet it is strictly the higher-variance route (Proposition~\ref{prop:dom})—a product rather than a sum. This variance is therefore not a flaw of any particular single-agent-reduction algorithm but an inherent cost of aggregating on the ratio side, and it can be removed by moving the aggregation into the advantage. The field's revealed preference for local advantages over local ratios is exactly what this variance-ordering predicts; we subsume the local-advantage heuristics as particular advantage-support choices and provide the design principle behind them.

\section{Discussion}\label{app:discussion}

\paragraph{Existing methods on the bias--variance plane.} The two supports provide a common coordinate system for methods that otherwise appear unrelated. IPPO sits at the low-variance, high-bias corner ($\SA=\SR=I$): both supports per-agent, so its estimator has minimal variance but ignores every externality. MAPPO moves the advantage to the team ($\SA=\one\one^\top,\SR=I$): it removes the missing-coupling bias but, on a large team, pays the additive variance of summing all rewards—and, as our $196$-intersection experiment shows, this already fails when the team is large and only a local neighborhood is truly coupled. Neighborhood-advantage methods (spatial-discount and local-critic actor--critics) occupy the principled middle: $\SA$ is a $k$-hop neighborhood, $\SR=I$, which our analysis identifies as the MSE-optimal choice when the coupling has range $k$. At the far corner lie single-agent reductions that treat $\prod_i\pi^i$ as one multi-action policy—for example, CMAT \citep{cmat2025}, which factorizes the joint policy through a latent consensus and then optimizes it with single-agent PPO on the joint action. Such methods effectively push both supports to fully joint ($\SA=\SR=\one\one^\top$), so $\tS=n\,\one\one^\top$—the maximal effective support. In our terms, this is the least biased and most variant point in the plane: it aggregates every neighbor's reward and multiplies every neighbor's ratio. Our results indicate that the first is unnecessary beyond the coupling neighborhood and the second is strictly harmful; the elegant equivalence to single-agent RL is real at the level of the expected objective, but it is bought with a variance that, in a large system, need not pay off. This is not a defect of any particular such algorithm but an inherent property of the reduction—the maximal-support corner is high-variance by construction—and our canonical form shows the variance is removable without changing what is learned, by moving all aggregation onto the advantage and sizing it to the coupling graph. This reframes per-agent versus centralized not as a binary but as a point in a two-dimensional support plane whose optimal location the theory pins down.

\paragraph{Does the result apply to multi-action (centralized) PPO?} A natural concern is that our analysis concerns multi-agent PPO, whereas a centralized controller may treat the whole team as a single agent whose action is the joint $\bm a=(a^1,\dots,a^n)$ and run vanilla single-agent PPO on it—the multi-action view underlying CMAT \citep{cmat2025}. Our results apply directly to this setting, because the only structural assumption we use is that the policy factorizes across action components, $\pi(\bm a)=\prod_i\pi^i(a^i)$ (Assumption~\ref{as:indep})—i.e., independent categorical heads, which is exactly how such multi-action policies are built. Under that factorization, multi-action PPO is the corner $(\SA,\SR)=(\one\one^\top,\one\one^\top)$: a single scalar advantage (the team return) is broadcast to every action component, and the importance weight is the ratio of the joint action, $\prod_i\varrho^i$. The canonical form (Theorem~\ref{thm:canon}) therefore governs it unchanged—the joint ratio is gradient-redundant relative to a per-agent ratio with a reweighted advantage, and strictly higher-variance (Prop.~\ref{prop:dom})—so the design rule transfers verbatim: clip each action component's ratio separately rather than as one $n$-way product, which turns multi-action PPO back into a per-agent-ratio scheme (a MAPPO-style update) at no cost to the expected gradient. Two boundaries are worth stating. First, the factorization is essential: an autoregressive joint head (component $a^{j}$ conditioned on $a^{<j}$) uses a different chain rule and is out of scope. Second, when the environment provides only a single non-decomposable team reward, the advantage cannot be localized (there is no per-agent reward to aggregate over a neighborhood), so the bias-reduction half of our rule is inapplicable; however, the variance half still holds—the joint ratio should be replaced by a per-agent ratio regardless—which is precisely the sense in which a MAPPO-style per-agent ratio dominates a fully centralized multi-action update even in the single-team-reward regime.

\paragraph{Two caveats the experiments make precise.} The design rule is exact at the level of the expected gradient, but two qualifications govern how visibly it matters in training. First, the advantage support reshapes the learned policy only when the coupling externality is not already internalized by an agent's own reward: in the directed dilemma, independent advantage collapses to the tragedy equilibrium and a neighborhood advantage recovers the social optimum, whereas in symmetric-payoff, congestion, and block games, independent advantage is already near-unbiased and larger supports only add variance (Fig.~\ref{fig:advsweep}). Second, the ratio's variance penalty is gated by the per-update policy shift $\chi^2$: a conservative trust region keeps $\chi^2$ so small that per-agent and joint ratios are indistinguishable even at hundreds of agents, and the penalty surfaces only as updates are driven off-policy (Fig.~\ref{fig:scaling}). Neither qualification weakens the design rule—the ratio's benefit is exactly zero and its cost is non-negative throughout—but both explain why the folklore that per-agent ratios are more stable is only intermittently observed in practice.

\paragraph{Implications for large-scale deployment.} The design principle is most consequential precisely where cooperative MARL is scaling: systems with hundreds to thousands of agents and local physical coupling—traffic-signal grids, sensor and robot fleets, power and communication networks, warehouse logistics. Two concrete recommendations follow. \emph{(1) Never aggregate the ratio.} At these scales, the joint ratio's variance factor $(1+\chi^2)^n$ is enormous the moment any batch drifts off-policy, and our $196$-intersection result shows the analogous failure for an over-large advantage (a global team advantage never learns, while local advantages cut queues by $2.4\times$ and waiting time by $3.7\times$). Keep the ratio per-agent; put every bit of cross-agent aggregation in the advantage, where the cost is additive. \emph{(2) Size the advantage to the physical coupling.} Because these systems come with a known interaction graph (the road network, the communication topology, the spatial adjacency), the MSE-optimal advantage support is directly available—the $k$-hop neighborhood for a coupling of range $k$—rather than something to be tuned blindly. This turns a hyperparameter search into a modeling choice and connects to the scalable-MARL line \citep{qu2020scalable,qu2022scalable}, whose exponential-decay property is exactly the condition under which a small advantage neighborhood is near-unbiased. Neither recommendation requires changing the network architecture or the training loop—only which rewards and which ratios enter each agent's update—so they are immediately actionable for existing MAPPO and IPPO codebases.

\paragraph{Future work: learning the support.} Our design rule assumes the coupling neighborhood is known—true for physically networked systems, but not in general. When the coupling graph is unknown or state-dependent, the open problem is to learn the right advantage support: which neighbors to aggregate, and at what radius, so as to minimize gradient MSE. This is a bias--variance model-selection problem—estimate each agent's influence graph (e.g., from reward sensitivities or attention over other agents' states) and include an agent in the support when its estimated coupling outweighs the variance its reward adds. A dynamic, per-state support that grows in strongly-coupled regimes and shrinks otherwise would sit at the MSE optimum adaptively, generalizing the fixed $k$-hop rule; the canonical form guarantees that whatever support is learned, it should be realized on the advantage and never on the ratio.

\end{document}